\documentclass{article}

\usepackage[usenames,dvipsnames,dvips]{color}
\definecolor{MyDarkBlue}{rgb}{0,0.08,0.45}
\usepackage[backref=none]{hyperref}
\hypersetup{pdfborder={0 0 0},
  colorlinks,
  urlcolor={MyDarkBlue},
  linkcolor={MyDarkBlue},
  citecolor={MyDarkBlue},
  breaklinks=true}

\usepackage{doi}
\providecommand{\eprint}[1]{}
\renewcommand{\eprint}[1]{arXiv:\href{http://arxiv.org/abs/#1}{#1}}

\usepackage[dvips]{color}

\usepackage{amsmath}
  \usepackage{paralist}
  \usepackage{graphics} 
  \usepackage{epsfig} 
\usepackage{graphicx}  \usepackage{epstopdf}
\hypersetup{urlcolor=blue, citecolor=red}

\usepackage{amsthm}

\newtheorem{theorem}{Theorem}[section]
\newtheorem{corollary}{Corollary}
\newtheorem{lemma}[theorem]{Lemma}
\newtheorem{proposition}{Proposition}

\theoremstyle{definition}

\newtheorem{remark}{Remark}
\newtheorem{example}{Example}

\providecommand{\subjclass}[1]{}
\providecommand{\keywords}[1]{}
\providecommand{\email}[1]{}
\providecommand{\thanks}[1]{}
\renewcommand{\thanks}[1]{}

\usepackage{latexsym}
\usepackage{bm}
\usepackage{upgreek}
\usepackage{mathrsfs}
\usepackage{times}
\usepackage{amsthm}
\usepackage{amssymb}
\usepackage{doi}
\usepackage{amsmath}

\DeclareOldFontCommand{\brianup}{\upshape}{\mathrm}

\DeclareSymbolFont{EUR}{U}{eur}{m}{n}
\SetSymbolFont{EUR}{bold}{U}{eur}{b}{n}
\DeclareSymbolFontAlphabet{\eur}{EUR}

\DeclareSymbolFont{EUB}{U}{eur}{b}{n}
\SetSymbolFont{EUB}{bold}{U}{eur}{b}{n}
\DeclareSymbolFontAlphabet{\eub}{EUB}

\DeclareSymbolFont{AMSb}{U}{msb}{m}{n}
\DeclareSymbolFontAlphabet{\mathbb}{AMSb}

\providecommand\ker{\hbox{\rm Ker\,}}
\renewcommand\ker{\mathop{\rm Ker}}
\newcommand{\scp}{A\sp 0}

\newcommand\eurA{\eur{A}}

\newcommand\eubA{\eub{A}}

\newcommand{\ech}{\mbox{\small\it q}}

\newcommand{\bme}{\eub{e}}
\newcommand{\bfN}{\mathrm{N}}
\newcommand\eurM{\eur{M}}

\newcommand{\range}{\mathop{\rm Range\,}}

\newcommand{\p}{\partial}
\renewcommand{\P}{\undefinednothing}
\newcommand{\mom}{P\hspace{-1.45ex}\slash\hspace{0.45ex}}
\newcommand{\at}[1]{\vert\sb{\sb{#1}}}

\def\R{\mathbb{R}}
\def\P{\mathbb{P}}
\def\A{\mathbb{A}}
\newcommand{\C}{\mathbb{C}}
\newcommand{\N}{\mathbb{N}}\newcommand{\Z}{\mathbb{Z}}
\newcommand{\abs}[1]{\vert #1 \vert}
\newcommand{\norm}[1]{\Vert #1 \Vert}
\newcommand{\sothat}{{\rm ;}\ }

\newcommand{\const}{\mathop{\rm const}}

\makeatletter\@addtoreset{equation}{section}
\makeatother

\renewcommand{\Re}{\mathop{\rm{R\hskip -1pt e}}\nolimits}
\renewcommand{\Im}{\mathop{\rm{I\hskip -1pt m}}\nolimits}

\title{Small amplitude solitary waves in the Dirac--Maxwell system}

\author{
{\sc Andrew Comech}
\\
{\small\it
Texas A\&M University, College Station, TX 77843, USA}
\\
{\small\it IITP, Moscow 127051, Russia}
\\~\\
{\sc David Stuart}
\\
{\small\it University of Cambridge, Cambridge CB3 0WA, UK}
}


 \email{comech@math.tamu.edu}
 \email{D.M.A.Stuart@damtp.cam.ac.uk}

\thanks{The first author was supported by the Russian Foundation
for Sciences (project 14-50-00150).}

\begin{document}
\maketitle


\bigskip




\begin{abstract}
We study nonlinear bound states, or solitary waves, 
in the Dirac--Maxwell system, proving the 
existence of solutions in which the Dirac wave function is of the form
$\phi(x,\omega)e^{-i\omega t}$, with
$\omega\in(-m,\omega\sb \ast)$ for some $\omega\sb \ast>-m$. The solutions
satisfy
$\phi(\,\cdot\,,\omega)\in H\sp 1(\R^3,\C^4)$, and
are small amplitude in the sense that
$\norm{\phi(\,\cdot\,,\omega)}^2\sb{L\sp 2}=O(\sqrt{m+\omega})$
and
$\norm{\phi(\,\cdot\,,\omega)}\sb{L\sp\infty}=O(m+\omega)$. The method of
proof is an implicit function theorem argument based on 
the identification of the nonrelativistic limit
as the ground state of the Choquard equation. This identification is in some ways
unexpected on account of the repulsive nature of the electrostatic interaction
between electrons, and arises as a manifestation of certain peculiarities
(Klein paradox)
which result from attempts to interpret the 
Dirac equation as a single particle quantum mechanical wave equation. 
\end{abstract}

\bigskip

\hfill
{\it To Vladimir Georgiev on the occasion of his 60th birthday}

\bigskip

\section{Introduction and results}
\label{sec-dm}

The Dirac equation, which appeared in \cite{dirac-1928}
just two years after the Schr\"odinger equation,
is the correct Lorentz-invariant equation to
describe particles with nonzero spin
when relativistic effects cannot be ignored.
The Dirac equation predicts accurately the 
energy levels of an electron in the Hydrogen atom, yielding
relativistic corrections to the spectrum of the Schr\"odinger equation.
Further higher order corrections arise on account of
electromagnetic self-interactions,
described mathematically by the
Dirac--Maxwell Lagrangian, 
which aims to provide
a self-consistent description of the dynamics
of an electron interacting with its own electromagnetic field.
The perturbative treatment of the Dirac--Maxwell
system in the framework of second quantization
allows computation of quantities such as the
energy levels and scattering cross-sections, which have been
compared successfully with experiment, although this quantum formalism
does not provide the type of tangible description of 
particles and dynamical processes familiar 
from classical physics.
Mathematically, the quantum theory (QED) has not been
constructed, and indeed may not exist in the accepted analytical sense.
In particular it is a curious fact that although the electron 
is the most stable elementary particle known to physicists today,
there is no mathematically precise formulation and proof of its
existence and stability.
This has resulted in an enduring interest in the classical Dirac--Maxwell
system, both in the physics and mathematics literature. Regarding
the former, the relevance of the {\em classical} equations of motion
for QED has been widely debated. The prevalent view
seems to be that the Dirac fermionic field does not have a direct 
meaning or limit in classical physics, and hence that the classical system is
not really directly relevant to the world of observation.
Nevertheless,
there have been numerous attempts,
by Dirac himself
as well as by many others -- see
\cite{MR0139402,wakano-1966,MR1364144} and references therein --
to construct localized solutions of the classical system, or some modification
thereof, with the aim of obtaining a more cogent mathematical description
of the electron (or other fundamental particles).
We consider the system of Dirac--Maxwell equations,
where the electron, 
described by the standard ``linear'' Dirac equation,
interacts with its own electromagnetic field
which is in turn required to obey the Maxwell equations:
\begin{equation}\label{dirac-maxwell-classical}
\begin{cases}
\gamma\sp\mu(i\p\sb\mu-\ech A\sb\mu)\psi-m\psi=0,
\\
\p\sp\mu\p\sb\mu A\sp\nu=J\sp\nu,
\quad \p\sb\mu A\sp\mu=0,
\qquad
0\le\mu,\nu\le 3,
\end{cases}
\end{equation}
with the charge-current density
$J\sp\mu=(\rho,\mathbf{J})\in\R\times\R^3$
generated by the spinor field itself:
\begin{equation}\label{current-density}
J\sp\mu=\ech\bar\psi\gamma\sp\mu\psi,
\qquad
0\le\mu\le 3.
\end{equation}
Above, $\rho$ and $\mathbf{J}$
are the charge and current, respectively.
We denote $\bar\psi=(\gamma\sp{0}\psi)\sp\ast=\psi\sp\ast\gamma\sp{0}$,
with
$\psi\sp\ast$
the hermitian conjugate of $\psi\in\C^4$.
The charge is denoted by $\ech$
(so that for the electron $\ech<0$);
the fine structure constant is the dimensionless coupling constant
$\alpha\equiv\frac{\ech^2}{\hbar c}\approx 1/137$.
We choose the units so that $\hbar=c=1$.
We have written the Maxwell equations using the Lorentz gauge condition
$\p\sb\mu A\sp\mu=0$.
The Dirac $\gamma$-matrices satisfy the
anticommutation relations
\[
\{\gamma\sp\mu,\gamma\sp\nu\}=2g\sp{\mu\nu},
\qquad
0\le\mu,\nu\le 3,
\]
with $g\sp{\mu\nu}=\mathop{\rm diag}[1,-1,-1,-1]$.
The four-vector potential $A\sp\mu$ has components 
$({\scp},\mathbf{A})$,
with
$\mathbf{A}=\{A\sp j\}_{j=1}^{3}$,
so that the lower index version
$A\sb\mu=g\sb{\mu\nu}A\sp\nu$ has components
$({\scp},-\mathbf{A})$ so $A\sb 0={\scp}$.
Following
\cite{MR0441102} and \cite{MR0187641},
we define the Dirac $\gamma$-matrices by
\begin{equation}
\gamma\sp{j}
=\left(
\begin{matrix} 0&\sigma\sb{j} \\ -\sigma\sb{j}&0\end{matrix}
\right),
\qquad
\gamma\sp{0}
=\left(
\begin{matrix} I\sb{2}&0 \\ 0&-I\sb{2}\end{matrix}
\right),
\end{equation}
where $I\sb{2}$ is the $2\times 2$ unit matrix
and
$\sigma\sb{j}$ are the Pauli matrices:
$\sigma\sb{1}=\left(\begin{matrix} 0&1 \\ 1&0\end{matrix}\right)$,
$\sigma\sb{2}=\left(\begin{matrix} 0&-i \\ i&0\end{matrix}\right)$,
$\sigma\sb{3}=\left(\begin{matrix} 1&0 \\ 0&-1\end{matrix}\right)$.
After introduction of a space-time splitting,
the system  \eqref{dirac-maxwell-classical}
takes the form
\begin{eqnarray}\label{stationary-eqns}
\begin{cases}
i\p\sb t\psi
=\bm\alpha\cdot(-i\bm\nabla-\ech \mathbf{A})\psi
+m\beta\psi
+\ech A\sp 0\psi,
\\
(\p\sb t^2-\Delta)A\sp 0=\ech\psi\sp\ast\psi,
\\
(\p\sb t^2-\Delta)\mathbf{A}=\ech\psi\sp\ast\bm\alpha\psi.
\end{cases}
\end{eqnarray}
Above,
$\bm\alpha=(\alpha\sp 1,\alpha\sp 2,\alpha\sp 3)$;
$\alpha\sp j$ and $\beta$ are the $4\times 4$ 
Dirac matrices:
\begin{equation}
\alpha\sp j
=\left(
\begin{matrix} 0&\sigma\sb{j} \\ \sigma\sb{j}&0\end{matrix}
\right),
\qquad
\beta
=\left(
\begin{matrix} I\sb{2}&0 \\ 0&-I\sb{2}\end{matrix}
\right)\,.
\end{equation}
The $\alpha$-matrices and $\gamma$-matrices are related by
\[
\gamma\sp j=\beta\alpha\sp j,
\quad
1\le j\le 3;
\qquad
\gamma\sp{0}=\beta.
\]

Numerical evidence for the existence of
solitary wave solutions
to the Dirac--Maxwell system \eqref{dirac-maxwell-classical}
was obtained in \cite{wakano-1966}
and then in \cite{MR1364144},
where it was suggested that
such solutions are produced
by the Coulomb repulsion
from the negative part of the essential spectrum
(the Klein paradox).
The numerical results of \cite{MR1364144}
showed that
the Dirac--Maxwell system
has infinitely many families of
\emph{solitary waves}
$\phi\sb{N}(x,\omega)e^{-i\omega t}$,
$\omega\gtrsim -m$.
Here the nonnegative integer $N$
denotes the number of nodes of the positronic component
of the solution
(number of zeros of the corresponding
spherically symmetric
solution to the Choquard equation; see
\S\ref{sect-nr}).
A variational proof of existence of solitary waves
for $\omega\in(-m,0)$
and with $N=0$
first appeared in
\cite{MR1386737},
and the generalization to handle $\omega\in(-m,m)$
is in \cite{MR1618672}.
%

In the present paper, we give a proof of existence
of solitary wave solutions to the Dirac--Maxwell system
based on the perturbation from the nonrelativistic limit
and also
obtain the precise asymptotics for the solution in this limit.

The solitary wave solution
$(\phi(x) e^{-i\omega t},A\sp\mu(x))$
satisfies the stationary system
\begin{equation}\label{omega-phi-is}
\omega\phi
=\bm\alpha\cdot(-i\bm\nabla-\ech \mathbf{A})\phi+m\beta\phi+\ech A\sp 0\phi,
\qquad
-\Delta A\sp\mu=\ech\bar\phi\gamma\sp\mu\phi.
\end{equation}

\begin{theorem}\label{theorem-sw-dm}
There exists ${\omega\sb\ast}>-m$ such that
for $\omega\in(-m,{\omega\sb\ast})$
there is a solution to
\eqref{omega-phi-is}
of the form
\[
\phi(x,\omega)=\begin{bmatrix}
\epsilon^3\varPhi\sb 1(\epsilon x,\epsilon)
\\
\epsilon^2\varPhi\sb 2(\epsilon x,\epsilon)
\end{bmatrix},
\qquad
\epsilon=\sqrt{m^2-\omega^2},
\]
with
\[
\varPhi
=\begin{bmatrix}\varPhi\sb 1\\\varPhi\sb 2\end{bmatrix}
\in C^\infty\bigl(
(0,{\epsilon\sb\ast})\,;\,
H^2(\R^3;\C^2)\oplus H^2(\R^3;\C^2)
\bigr),
\qquad
{\epsilon\sb\ast}=\sqrt{m^2-{\omega\sb\ast}^2},
\]
and with
\[
A\sp\mu\in C^\infty\bigl(
(0,{\epsilon\sb\ast})\,;\, \dot H\sp 1(\R^3,\R)\cap L^\infty(\R^3,\R)\bigr),
\qquad
0\le \mu\le 3.
\]
Above,
$\dot H^1=\dot H^1(\R^3,\R)$ is the homogeneous Dirichlet
space of $L^6$ functions with the norm
\[
\|f\|_{\dot H^1}^2:=\int\sb{\R^3}|\nabla f|^2\,dx<\infty.
\]
For small $\epsilon>0$, one has
\[
\norm{{\varPhi\sb{1}}-\hat\varPhi\sb{1}}\sb{H\sp 2}
+
\norm{{\varPhi\sb{2}}-\hat\varPhi\sb{2}}\sb{H\sp 2}
=O(\epsilon^2),
\]
where
$\hat\varPhi\sb 1(y)$, $\hat\varPhi\sb 2(y)$
are of Schwartz class.
The solutions can be chosen so that
in the nonrelativistic limit $\epsilon=0$
one has
\begin{equation}
\hat\varPhi\sb 2(y)=\varphi_0(y)\bm{n},
\qquad
\hat\varPhi\sb 1(y)=\frac{i}{2m}
\bm\sigma\cdot\bm\nabla\hat\varPhi\sb 2(y),
\end{equation}
where $\bm{n}\in\C^2$, $\abs{\bm{n}}=1$,
and $\varphi_0\in\mathscr{S}(\R^3)$
is a strictly positive, spherically symmetric,
strictly monotonically decaying (as a function of $\abs{y}$)
solution of Schwartz class
to the Choquard equation
\begin{equation}\label{phi0}
-\frac{1}{2m}\varphi=-\frac{1}{2m}\Delta\varphi
-\Big(\frac{\ech^2}{4\pi\abs{\,\cdot\,}}\ast\varphi^2\Big)\varphi,
\qquad
\varphi(y)\in\R,
\quad
y\in\R^3.
\end{equation}
The Dirac field \(\phi(x,\omega) \) has exponential decay in $x$,
while the electromagnetic potential
satisfies
\[
A\sp 0(x,\omega)
\,=\,\ech\,\frac{\norm{\phi}_{L^2}^2}
              {4\pi|x|}\,+\,O(\langle x\rangle^{-2})\,,\qquad
 \mathbf{A}(x,\omega)\,=\,O(\langle x\rangle^{-2})             
\]
as \(|x|\to+\infty\,. \)
\end{theorem}

\begin{remark}
The existence of a positive spherically-symmetric solution
$\varphi_0\in\mathscr{S}(\R^3)$
to \eqref{phi0}
was proved in 
\cite{MR0471785}.
\end{remark}

To prove Theorem~\ref{theorem-sw-dm},
we construct solitary wave solutions by deforming
the solutions to the nonrelativistic limit
(represented by the Choquard equation)
via the implicit
function theorem.
Such a method was employed
in \cite{MR1750047,2008arXiv0812.2273G,MR3670258}
for the nonlinear Dirac equation
and in \cite{MR2593110,MR2647868,MR2671162}
for Einstein--Dirac and Einstein--Dirac--Maxwell systems.

One motivation for presenting here a new existence proof
for the Dirac--Maxwell solitary waves
is to realize mathematically
the physical intuition sketched in \cite{MR1364144}, which explains
the existence of these bound state solutions in terms of the Klein paradox
(see e.g. \cite[\S3.3]{MR0187641}).
Moreover, once one knows
that the excited eigenstates of the Choquard equation
are nondegenerate
(currently this nondegeneracy
is established only for the ground state, $N=0$ \cite{MR2561169}),
our argument will yield the existence
of excited
solitary wave solutions in the Dirac--Maxwell system,
extending the results of \cite{MR1386737}
to $N\ge 1$ (see Remark~\ref{remark-variational} below);
as mentioned in that article,
the variational methods are hard to
generalize to prove the existence of multiple solitary
waves for each $\omega$
in the Dirac--Maxwell system
(although such a multiplicity result has been obtained 
in \cite{MR1386737} for the Dirac--Klein--Gordon system).

Another motivation for the bifurcation approach
is that
having the nonrelativistic (or small-amplitude)
asymptotics of solitary waves
is the first step towards analyzing their stability.
Indeed, the physical significance
of solitary waves
requires not only existence but 
also stability, and it is to be hoped that the type of 
detailed information about the solutions
which is a consequence of the existence proof in this article, but
does not seem to be so easily accessible from the original variational 
constructions, will be helpful in future stability analysis
(see Remark~\ref{remark-stability} below).
In this context we mention some recent stability results for
the nonlinear Dirac equation.
Numerical results suggest that
the nonlinear Dirac equation with scalar-type
self-interaction, known as the Soler model,
possesses solitary waves which are spectrally stable,
that is,
the linearization at the solitary wave
has purely imaginary spectrum.
Numerics indicate that all solitary waves
in the cubic Soler model in one spatial dimension
(known as Gross--Neveu model)
are spectrally
stable, except perhaps for $\omega$ very close to $\omega=0$
\cite{MR2892774}.
Numerically,
for the Soler model with cubic nonlinearity
in the two-dimensional case,
there exists $\omega_0\in(0,m)$
such that solitary waves with
$\omega\in(\omega_0,m)$ are spectrally stable,
and the same is expected in the three-dimensional case
with $\omega\lesssim \omega_{\ast}\approx 0.936 m$
\cite{PhysRevLett.116.214101}.
More general results on the spectral stability
(that is, absence of linear instability)
in the nonlinear Dirac equation
have been obtained in \cite{MR3311594}
(possibility of bifurcations of nonzero-real-part eigenvalues
from the origin)
and in \cite{MR3530581}
(possibility of bifurcations of nonzero-real-part eigenvalues
from the essential spectrum).
The spectral stability of small amplitude solitary waves
(that is, solitary waves in the nonrelativistic limit)
in the charge-subcritical and charge-critical cases,
when the nonlinear term is $\abs{\psi\sp\ast\beta\psi}^k\beta\psi$
with $k\le 2/n$,
where $n\ge 1$ is the spatial dimension
(under the technical assumption
that $k>k_n$, with some $k_n\in(0,2/n)$)
was proved in \cite{linear-b}.
The linear instability of small amplitude solitary waves
in charge-supercritical case
$k>2/n$ (with $k<2/(n-2)$ if $n\ge 3$)
was shown in \cite{MR3208458}
(the general case of non-integer $k$
also needs the proof of
existence of solitary waves from \cite{MR3670258}).
The asymptotic stability of solitary waves
in the one-dimensional Soler model with respect to
``radially symmetric'' perturbations
has been proved in \cite{gn-stability}.

Here is the plan of the paper.
We give the heuristics and expected nonrelativistic scaling in \S\ref{sect-heur}.
The Choquard equation,
which is the nonrelativistic limit of the Dirac--Maxwell system,
is considered in \S\ref{sect-nr}.
In \S\ref{sect-exist}, we complete
the proof of existence of solitary waves
via the implicit function theorem.

\section{Heuristics on the nonrelativistic limit}
\label{sect-heur}
The small amplitude waves constructed in
Theorem~\ref{theorem-sw-dm}
are best understood physically  in terms of the non-relativistic limit. Since we
have set the speed of light and other physical constants equal to one, the 
relevant small parameter is the excitation energy (or frequency) as compared 
to the mass $m$.
To develop some preliminary intuition regarding the non-relativistic limit, 
following \cite{MR1364144}, we neglect the magnetic field
described by the vector-potential $A\sp j$,
getting
\[
i\p\sb t\psi=-i\bm\alpha{\cdot\bm\nabla}\psi
+m\beta\psi
+\ech A\sp 0\psi,
\qquad
(\p\sb t^2-\Delta)A\sp 0=\ech\psi\sp\ast\psi\,.
\]
We consider a family of solitary waves
\[
\psi(x,t)=\phi(x,\omega)e^{-i\omega t},
\qquad
\omega\in\R,
\]
with
$
\phi(x,\omega)
=
\begin{bmatrix}
\phi\sb 1(x,\omega)\\\phi\sb 2(x,\omega)
\end{bmatrix}\in\C^4
$,
where
$\phi\sb 1(x,\omega),\,\phi\sb 2(x,\omega)\in\C^2$
and $A\sp 0(x,\omega)\in\R$.
Then $\phi\sb 1$, $\phi\sb 2$, and $A\sp 0$ satisfy
\begin{eqnarray}
\label{cs}
\begin{cases}
(\omega-m)\phi\sb 1
=-i\bm\sigma{\cdot\bm\nabla}\phi\sb 2
+\ech A\sp 0\phi\sb 1,
\\
(m+\omega)\phi\sb 2
=-i\bm\sigma{\cdot\bm\nabla}\phi\sb 1
+\ech A\sp 0\phi\sb 2,
\\
-\Delta A\sp 0=\ech
(\phi\sb 1\sp\ast\phi\sb 1+\phi\sb 2\sp\ast\phi\sb 2)\,,
\end{cases}
\end{eqnarray}
where $\bm\sigma=(\sigma\sb 1,\sigma\sb 2,\sigma\sb 3)$,
the vector formed from the Pauli matrices.
Let us try to find small amplitude solitary waves with $\omega\gtrsim -m$.
Then $A\sp 0$ is small and
$-2m\phi\sb 1\approx -i\bm\sigma{\cdot\bm\nabla}\phi\sb 2$,
\begin{eqnarray}\label{the-above}
-(m+\omega)\phi\sb 2
\approx
-\frac{1}{2m}\Delta\phi\sb 2
-\ech A\sp 0\phi\sb 2,
\qquad
-\Delta A\sp 0
\approx
\ech
(\phi\sb 1\sp\ast\phi\sb 1+\phi\sb 2\sp\ast\phi\sb 2).
\end{eqnarray}
We notice from the above system
that the effect of the electromagnetic interaction is now
attractive; this is because we are analyzing states which bifurcate from
the negative energy spectrum.
It can easily be seen that if the same
reasoning as above is applied when \(\omega\lesssim m \),
then it leads to an equation
with a repulsive interaction, as is normal in electrostatics. This
generation of an effective attraction out of negative energy states is one of
a number of curious phenomena which arise from attempts to treat the
Dirac equation as a single particle wave equation, collectively
referred to under the label ``Klein paradox''
\cite{sakurai1967advanced}.

Let
$\epsilon^2=m^2-\omega^2$, $0<\epsilon\ll m;\,$
then \eqref{the-above} suggests the following scaling:
\begin{eqnarray}\label{scaling}
&
y=\epsilon x,
\qquad
\p\sb x=\epsilon\p\sb y,
\qquad
A\sp 0(x,\omega)=\epsilon^2 \eurA\sp 0(\epsilon x,\epsilon),
\nonumber
\\
&
\phi\sb 1(x,\omega)=\epsilon^{3}\varPhi\sb 1(\epsilon x,\epsilon),
\qquad
\phi\sb 2(x,\omega)=\epsilon^{2}\varPhi\sb 2(\epsilon x,\epsilon).
\end{eqnarray}
Note that
while $\phi\sb j$ and $A\sp 0$ depend on $\omega$ and $x$,
it is convenient to consider the scaled
functions
$\eurA\sp 0$ and $\varPhi_j$ as functions of
$y=\epsilon x$ and $\epsilon=\sqrt{m^2-\omega^2}$.
In the limit $\epsilon\to 0$,
denoting
\[
\hat\varPhi=\lim\sb{\epsilon\to 0}\varPhi,
\qquad
\hat\eurA\sp 0=\lim\sb{\epsilon\to 0}\eurA\sp 0,
\]
we arrive at the system
\begin{equation}\label{phi1phi2}
\begin{cases}
-2m\hat\varPhi\sb 1
=-i\bm\sigma{\cdot\bm\nabla}\sb{y}\hat\varPhi\sb 2,
\\
-\frac{1}{2m}\hat\varPhi\sb 2
=-i\bm\sigma{\cdot\bm\nabla}\sb{y}\hat\varPhi\sb 1
+\ech \hat\eurA\sp 0\hat\varPhi\sb 2,
\\
-\Delta\sb y \hat\eurA\sp 0=\ech\hat\varPhi\sb 2\sp\ast\hat\varPhi\sb 2\,,
\end{cases}
\end{equation}
which
can be rewritten as the following equation for $\hat\varPhi_2(y)$
and $\hat\eurA\sp 0(y)$ only:
\begin{equation}\label{phi2}
-\frac{1}{2m}\hat\varPhi\sb 2
=-\frac{1}{2m}\Delta\sb y\hat\varPhi\sb 2
-\ech \hat\eurA\sp 0\hat\varPhi\sb 2,
\qquad
-\Delta\sb y \hat\eurA\sp 0=\ech\hat\varPhi\sb 2\sp\ast\hat\varPhi\sb 2\,,
\end{equation}
with the understanding that $\hat\varPhi_1(y)$
is then obtained from the first equation
of \eqref{phi1phi2}.

\begin{remark}
Regarding self-consistency of this approximation: one can check that,
when using the scaling \eqref{scaling},
the magnetic field vanishes to higher order
in the limit $\epsilon\to 0$,
in agreement with \cite{MR1364144}.
Indeed,
$
\mathbf{A}
=-\Delta^{-1}\mathbf{J},
$
where
$\mathbf{J}=\ech\psi\sp\ast\bm\alpha\psi=O(\epsilon^5)$,
hence $\mathbf{A}=-\Delta^{-1}\mathbf{J}=O(\epsilon^3)$.
The second equation from \eqref{cs}
would then take the form
\[
(m+\omega)\phi\sb 2
=-i\bm\sigma{\cdot\bm\nabla}\phi\sb 1-\ech {\mathbf{A}\cdot}\bm\sigma\phi\sb 1
+\ech A\sp 0\phi\sb 2,
\]
where
${\mathbf{A}\cdot}\bm\sigma\phi\sb 1=O(\epsilon^6)$
while other terms are $O(\epsilon^4)$. Thus the approximation is
at least formally self-consistent; the analysis in \S\ref{sect-exist}
makes this rigorous. 
\end{remark}

\begin{remark}
Regarding the symmetry: while
it is clear that radial symmetry
of both $\phi_1$ and $\phi_2$ is inconsistent with \eqref{phi1phi2},
solutions of the form given in \cite{wakano-1966},
\begin{equation}\label{ansatz}
\phi\sp{\mathrm{I}}(x)=
\begin{bmatrix}
g(r)\begin{bmatrix}
1\\0
\end{bmatrix}
\\
i f(r)
\begin{bmatrix}\cos\theta\\ e^{i\upphi}\sin\theta\end{bmatrix}
\end{bmatrix},
\qquad
\phi\sp{\mathrm{II}}(x)=
\begin{bmatrix}
f(r)
\begin{bmatrix}\cos\theta\\ e^{i\upphi}\sin\theta\end{bmatrix}
\\
i g(r)\begin{bmatrix}
1\\0
\end{bmatrix}
\end{bmatrix},
\end{equation}
are permitted in principle,
suggesting that in the non-relativistic limit $\hat\varPhi_2$ 
could be radial, or, to be more precise,
of the form
$
\hat \varPhi\sb 2(y)
=\varphi(y)\begin{bmatrix}1\\0\end{bmatrix}\in\C^2,
$
where
the spherically symmetric function
$\varphi(y)\in\C$, $y\in\R^3$ is to satisfy
\begin{equation}\label{phiphi}
-\frac{1}{2m}\varphi
=-\frac{1}{2m}\Delta\sb y\varphi
-\ech \hat\eurA\sp 0\varphi,
\qquad
-\Delta\sb y \hat\eurA\sp 0=\ech|\varphi|^2\,.
\end{equation}
The starting point for our perturbative construction of solitary
wave solutions to \eqref{stationary-eqns} is indeed 
a radial solution of \eqref{phiphi}, although the exact form of these
solitary waves
has to be modified from \eqref{ansatz} when the 
effect of the magnetic field 
$\mathbf{B}=\nabla\times\mathbf{A}$
is included; see \cite[\S5]{MR1364144}. The method of proof
we employ does not require any particular symmetry class of the
solitary wave.
\end{remark}

The above discussion suggests  that the system \eqref{phiphi} 
determines the non-relativistic limit in the leading order.
The system \eqref{phiphi} describes a Schr\"odinger wave function with an
attractive self-interaction determined by the Poisson equation.
Since the sign of the interaction is attractive,
\eqref{phiphi} is often referred to as the stationary Newton--Schr\"odinger system.
It is equivalent to a nonlocal equation for $\varphi$ known as the Choquard
equation, which is the subject of the next section.

\section{The nonrelativistic limit: the Choquard equation}
\label{sect-nr}

One arrives at the system \eqref{phiphi}
when looking for solitary wave solutions in the system
\begin{equation}\label{nrl}
\begin{cases}
i\p\sb t\psi=-\frac{1}{2m}\Delta\psi
-\ech V\psi,
\\
-\Delta V=\ech\psi\sp\ast\psi,
\end{cases}
\qquad
\psi(x,t)\in\C,
\qquad
V(x,t)\in\R,
\qquad
x\in\R^3.
\end{equation}
This is the time-dependent Newton--Schr\"odinger system.
If $\big(\varphi(x,\omega) e^{-i\omega t},V(x,\omega)\big)$
is a solitary wave solution, then
$\varphi$ and $V$ satisfy the stationary system
\begin{equation}\label{nss}
\omega\varphi=-\frac{1}{2m}\Delta\varphi-\ech V\varphi,
\qquad
-\Delta V=\ech\abs{\varphi}^2.
\end{equation}
We rewrite the system \eqref{nrl} in the non-local form,
which is known as the {\em Choquard equation} \cite{MR0471785}:
\begin{equation}\label{choquard}
i\p\sb t\psi=-\frac{1}{2m}\Delta\psi+\ech^2\Delta^{-1}(\abs{\psi}^2)\psi,
\qquad
\psi(x,t)\in\C,
\qquad
x\in\R^3,
\end{equation}
where
$\Delta^{-1}$ is the operator of convolution with $-\frac{1}{4\pi \abs{x}}$.
The solitary waves are solutions of the form
$
\psi(x,t)=\varphi(x,\omega)e^{-i\omega t},
$
with $\varphi$ satisfying the non-local scalar equation
\begin{equation}\label{lambda-phi}
\omega\varphi
=-\frac{1}{2m}\Delta \varphi+\ech^2\Delta^{-1}(\abs{\varphi}^2)\varphi\,.
\end{equation}
This suggests the following variational formulation for the problem:
find critical points of
\begin{equation}\label{def-E-Choquard}
{E}(\varphi)
=\frac{1}{2m}\int\sb{\R^3} |\nabla\varphi|^2\,dx
-\frac{\ech^2}{8\pi}\int\sb{\R^3\times\R^3}
\frac{|\varphi(x)|^2|\varphi(y)|^2}{|x-y|}\,dx\,dy\,,
\end{equation}
subject to the constraint ${{Q}}(\varphi)=\const$, with
the charge functional defined by
\begin{eqnarray}\label{def-q0}
{{Q}}(\varphi)=\int\sb{\R^3}\abs{\varphi(x)}^2\,dx.
\end{eqnarray}
This formulation
is the basis of the existence and uniqueness proofs in the references
which are summarized in the following theorem.

\begin{lemma}[\cite{MR0471785,MR591299,MR2592284}]
For all $\omega<0$
and $N\in\Z$, $N\ge 0$,
the equation \eqref{choquard}
admits solitary wave solutions
\[
\psi(x,t)=\varphi_N(x,\omega)e^{-i\omega t},
\qquad
\lim\sb{\abs{x}\to\infty}\varphi_N(x,\omega)
=0,
\]
with $\varphi_N(x,\omega)$
a spherically symmetric solution of
\eqref{lambda-phi}.
These solutions
differ by the number $N$ of zeros (or nodes),
of the profile functions
$\varphi_N(x,\omega)$,
considered as functions of $r=\abs{x}$.
The profile function $\varphi_0$ 
with no zeros minimizes the value of the energy functional
${E}(\varphi)$
amongst functions
with fixed $L^2$ norm,
and is the unique (up to translation) positive $H^1$ solution of 
\eqref{lambda-phi}; the corresponding solitary wave
is called the \emph{ground state}.
\end{lemma}


\begin{remark}\label{remark-variational}
Together with the heuristics in the previous section,
the above result suggests that
for $\omega$ sufficiently close to $-m$
there might exist
infinitely many
families
of solitary waves
to the Dirac--Maxwell system,
which differ by the number of nodes
of the \emph{positronic} component
(two lower components of $\phi$).
\end{remark}

\begin{remark}\label{remark-scaling}
The $\varphi(x,\omega)$ and $V(x,\omega)$ for different
values of $\omega<0$ can be scaled to produce a standard form
as follows. Let $\zeta>0$ be such that
$\omega=-\zeta^2$ and write
$y=\zeta x\,,$ $\varphi(x,\omega)=\zeta^2 u(\zeta x)\,,$
and $V(x,\omega)=\zeta^2 v(\zeta x)$. Then \eqref{nss}
is equivalent to  the following
system for $u(y)$, $v(y)$:
\begin{equation}\label{Phi-W}
-u=-\frac{1}{2m}\Delta\sb y u-\ech v u,
\qquad
-\Delta\sb y v=\ech\abs{u}^2\,.
\end{equation}
\end{remark}


In the remainder of this 
section we summarize the properties of the linearized
Choquard equation which follow from \cite{MR2561169} and
are needed in \S\ref{sect-exist}.
Consider a solution to the
Choquard equation of the form
\[
\psi(x,t)=\big(\varphi_0(x)+R(x,t)+i S(x,t)\big)e^{-i{\omega\sb 0} t},
\]
with $R(x,t)$, $S(x,t)$ real-valued.
The linearized equation for $R$, $S$
is:
\begin{eqnarray}\label{lin}
&
\p\sb t
\begin{bmatrix}R\\S\end{bmatrix}
=
\begin{bmatrix}0&L\sb 0\\-L\sb 1& 0\end{bmatrix}
\begin{bmatrix}R\\S\end{bmatrix},
\end{eqnarray}
where
\begin{eqnarray}
\label{def-l0-l1}
L\sb 0=-\frac{1}{2m}\Delta-{\omega\sb 0}
+\ech^2\Delta^{-1}(\varphi_0^2),
\qquad
L\sb 1=L\sb 0+2\ech^2\Delta^{-1}(\varphi_0\,\cdot\,)\varphi_0.
\end{eqnarray}
Notice that
$L\sb 1=\frac{1}{2}\bigl(E''(\varphi_0)-{\omega\sb 0} {{Q}}''(\varphi_0)
\bigr)$,
with $E(\varphi)$
from \eqref{def-E-Choquard}.
Both $L\sb 0$ and $L\sb 1$ are unbounded operators
$L^2(\R^3)\to L^2(\R^3)$
which are self-adjoint with domain
$H^2(\R^3)\subset L^2(\R^3)$.

\begin{lemma}\label{lemma-l0-l1}
The self-adjoint operator $L\sb 0:H^2\to L^2$
is positive-definite,
with $0\in\sigma\sb d(L\sb 0)$
a simple eigenvalue
corresponding to a positive eigenfunction
$\varphi_0$.
The range of
$L\sb 0$ is $\{\varphi_0\}^{\perp}$, the $L^2$-orthogonal complement of
the linear span of $\varphi_0$.

The self-adjoint operator $L\sb 1:H^2\to L^2$
has exactly one negative eigenvalue,
which we denote $-\Lambda\sb 0$, and
has a three-dimensional kernel $\ker L\sb 1$ 
spanned by $\{\partial_j \varphi_0\}_{j=1}^3$.
The range of
$L\sb 1$ is $(\ker L\sb 1)^{\perp}$, the $L^2$-orthogonal complement of
the linear span of the $\{\partial_j \varphi_0\}_{j=1}^3$.
\end{lemma}

\begin{proof}
Clearly $L\sb 0 \varphi_0=0$;
since $\varphi_0$ is positive,
it follows that $0$ is the lowest eigenvalue of $L\sb 0$
(which is thus non-degenerate),
with the rest of the spectrum separated from zero.

Now we focus on $L_1$;
we proceed similarly to \cite[Lemma 5.4.3]{kikuchi-thesis}.
The $N=0$ ground state solution $\varphi_0$ to \eqref{lambda-phi}
is characterized in \cite{MR0471785} as the solution,
unique up to translation and
phase rotation,
to the following minimization problem:
\begin{equation}\label{e-is-i}
E(\varphi_0)=I\sb\mu:=
\inf\big\{E({\varphi})\sothat\ {\varphi}\in H\sp 1(\R^3),
\ \norm{{\varphi}}\sb{L\sp 2}^2=\mu
\big\},
\end{equation}
for certain $\mu>0$;
above, $E(\varphi)$ is from \eqref{def-E-Choquard}.
We claim that this implies that $L\sb 1\geq 0$ on
$\{\varphi_0\}^{\perp}$. Indeed,
let $\norm{v}\sb{L\sp 2}=\norm{\varphi_0}\sb{L\sp 2}$,
$\langle v,\varphi_0\rangle=0$.
For $s\in(-1,1)$,
define ${\varphi}\sb s=(1-s^2)^{1/2}\varphi_0+s v$,
so that ${{Q}}({\varphi}\sb s)={{Q}}(\varphi_0)$.
Calculating that ${\varphi}\sb s\at{s=0}=\varphi_0$,
$\p\sb s\at{s=0}{\varphi}\sb s=v$,
$\p\sb s^2\at{s=0}{\varphi}\sb s=-\varphi_0$ we deduce from
\eqref{e-is-i}:
\begin{eqnarray}
\label{l1-positive}
&&
\hskip -0.3cm
0\le\p\sb s^2\at{s=0}E({\varphi}\sb s)
=\langle E'(\varphi_0),-\varphi_0\rangle
+\langle E''(\varphi_0)v,v\rangle
\\
\nonumber
&&
=
-{\omega\sb 0}\langle {{Q}}'(\varphi_0),\varphi_0\rangle
+\langle E''(\varphi_0)v,v\rangle
=
\langle v,(E''-{\omega\sb 0} {{Q}}'')v\rangle
=\,2\,
\langle v,L_1 v\rangle,
\end{eqnarray}
establishing the claim.
We took into account that
$\varphi_0$ satisfies
the stationary equation
$E'(\varphi_0)={\omega\sb 0} {{Q}}'(\varphi_0)$
and also that
\[
\langle {{Q}}'(\varphi_0),\varphi_0\rangle
=2\norm{\varphi_0}\sb{L\sp 2}^2
=2\norm{v}\sb{L\sp 2}^2
=\langle {{Q}}'(v),v\rangle
=\langle {{Q}}''v,v\rangle.
\]
So $L\sb 1$ is non-negative on a codimension one subspace.
On the other hand,
since the integral kernel of $\Delta^{-1}$
is strictly negative, while $\varphi_0$
is strictly positive and $L\sb 0 \varphi_0=0$, it follows that
$\langle \varphi_0,L\sb 1\,\varphi_0\rangle<0$ so that there
certainly exists one negative eigenvalue characterized as
\[
-\Lambda\sb 0
:=\inf\big\{\langle {v},L\sb 1 {v}\rangle\sothat\ \norm{{v}}\sb{L\sp 2}=1
\big\}<0.
\]
Let ${\eta\sb 0}$ be the corresponding eigenfunction,
$L\sb 1 {\eta\sb 0}=-\Lambda\sb 0 {\eta\sb 0}$. To prove that
$(-\Lambda\sb 0,0)\subset\rho(L\sb 1)$, which is the resolvent set of $L\sb 1$,
consider the minimization problem
\begin{equation}\label{mu-positive}
\inf\big\{
\langle v,L\sb 1 v\rangle\sothat\ \norm{v}^2=1,
\ \langle \eta\sb 0, v\rangle=0
\big\}\, .
\end{equation}
The relation
$L\sb 0 \varphi_0=0$, together with translation invariance,
implies that
$L\sb 1\p\sb {j}\varphi_0=0$.
Moreover,
it is proved in \cite{MR2561169} that
$\varphi_0$ is nondegenerate,
in the sense that
the kernel of $L\sb 1$ is spanned by the $\p\sb {j}\varphi_0$,
$1\le j\le 3$.
Hence,
by consideration of linear combinations of the eigenfunctions $\eta_0$ 
and $\p\sb {j}\varphi_0$, we conclude that
the value defined by \eqref{mu-positive} is $\leq 0$. In fact it must
equal zero since if it were negative a simple compactness argument of
the type appearing in \cite[Proof of Proposition 2.9]{MR783974},
based on the negativity of ${\omega\sb 0}$, would imply the existence
of a negative eigenvalue in the interval $(-\Lambda_0,0)$ and with 
the corresponding eigenfunction ${\eta\sb 1}$ orthogonal to ${\eta\sb 0}$.
But since
${\eta\sb 0},{\eta\sb 1}$ would then be an orthogonal pair of
eigenfunctions of $L\sb 1$ with negative eigenvalues, there would necessarily exist
some non-trivial linear combination of them having zero inner product with
$\varphi_0$, contradicting the fact that $L\sb 1$
is non-negative on $\{\varphi_0\}^{\perp}$ (cf. \eqref{l1-positive}).
\end{proof}

We will also need the following bounds for the inverses of \(L\sb 0 \) and
\(L\sb 1 \).
\begin{corollary}\label{bounds}
  \(L\sb 0^{-1} \) is a bounded operator $\{\varphi_0\}^{\perp}\cap L^2\to H^2$, while
  \(L\sb 1^{-1} \) is a bounded operator $(\ker L\sb 1)^{\perp}\cap L^2\to H^2$.
Also, in terms of the exponentially weighted Sobolev spaces \(H^{s,\theta}(\R^3) \),
with $s\in\N_0=\{0,1,2\dots\}$ and $\theta\ge 0$,
with the norms
\begin{eqnarray}\label{def-h-theta}
\|u\|_{H^{s,\theta}}\,=\,
\sum_{\alpha\in\N_0^3,\,\abs{\alpha}\le s}
\,\|e^{\theta |x|}\partial_x^\alpha u\|_{L^2(dx)}\,,
\end{eqnarray}
the mappings
\begin{eqnarray}
\nonumber
&&
L\sb 0^{-1}:\;\{\varphi_0\}^{\perp}\cap
  H^{s,\theta}\to H^{s+2,\theta}
\\
\nonumber
&&
L\sb 1^{-1}:\;(\ker L\sb 1)^{\perp}\cap
  H^{s,\theta}\to H^{s+2,\theta}\,
\end{eqnarray}
are bounded for \(\theta<|\omega_0|\,. \)

\end{corollary}

We conclude with a few remarks on the stability
of solitary waves to the Choquard equation.
By Remark~\ref{remark-scaling} we know the
$\omega$-dependence of
a localized solution
$\varphi(x,\omega)e^{-i\omega t}$
to \eqref{choquard}:
one has
$\varphi(x,\omega)=\zeta^2 u(\zeta\abs{x})$,
where $\zeta=\sqrt{-\omega}$. From this
we can obtain how the charge depends on the frequency $\omega<0$:
\[
{{Q}}(\omega)
=\int\limits\sb{\R^3}\abs{\varphi(x,\omega)}^2\,dx
=\zeta^4\int\limits\sb{\R^3}\abs{u(\zeta x)}^2\,dx
=\zeta\int\limits\sb{\R^3}\abs{u(y)}^2\,dy
=\abs{\omega}^{\frac 1 2}\int\limits\sb{\R^3}\abs{u(y)}^2\,dy.
\]
It follows that for all negative frequencies
one has
$
\frac{d}{d\omega}{{Q}}(\omega)<0\,.
$
By the Vakhitov--Kolokolov stability criterion \cite{VaKo},
this leads us to expect the spectral stability
of no-node solitary waves (the ground states)
in the Choquard equation.

\begin{proposition}\label{prop-stability-ns}
The ground state solitary wave
$\varphi_0(x)e^{-i{\omega\sb 0} t}$
of the Choquard equation \eqref{choquard}
is spectrally stable.
\end{proposition}

\begin{proof}
To determine the point spectrum of
$JL=
\begin{bmatrix}0&L\sb 0\\-L\sb 1&0\end{bmatrix}$
observe that if $\begin{bmatrix}R\\S\end{bmatrix}$ is an eigenfunction
corresponding to the eigenvalue $\lambda\in\C$,
then $-\lambda^2 R=L\sb 0 L\sb 1 R$.
If $\lambda\ne 0$,
then one concludes that
$R$ is orthogonal to $\ker L\sb 0$, which is the linear span of $\{\varphi_0\}$,
and hence we can apply $L\sb 0^{-1}$;
taking then the inner product with $R$,
we deduce that:
\[
-\lambda^2
\langle R,L\sb 0^{-1}R\rangle
=
\langle R,L\sb 1 R\rangle\,,
\]
which implies that $\lambda^2\in\R$.
Moreover,
since we already argued that \eqref{mu-positive} equals zero,
one has
$\lambda^2\le 0$,
which yields
$\sigma_d(JL)\subset i\R$
and hence the absence of
exponentially growing modes at the linearized level.
Let us mention that
the (nonlinear) orbital stability of the ground state solitary
wave solution to the Choquard equation
was proved in \cite{MR677997}.
\end{proof}

\begin{remark}\label{remark-stability}
In view of \cite{MR3208458,linear-b},
one expects that the spectral stability or linear instability
of small amplitude solitary waves
is directly related to the spectral stability or linear instability
of the corresponding nonrelativistic limit,
which for Dirac--Maxwell is given by the Choquard equation.
We hope that this may provide a route to understanding
stability of small solitary waves solutions for
the Dirac--Maxwell system.
\end{remark}

\section{Proof of existence of solitary waves in Dirac--Maxwell system}
\label{sect-exist}

In this section, we complete the proof of Theorem~\ref{theorem-sw-dm}.
It is
obtained as a consequence of Proposition~\ref{prop-sw-dm} after
the application of a rescaling motivated by the discussion in \S\ref{sect-nr}.

We write
$
\phi(x,\omega)
=
\begin{bmatrix}
\phi\sb 1(x,\omega)\\\phi\sb 2(x,\omega)
\end{bmatrix}\in\C^4
$,
where for $j=1,2$ the $\phi_j\in\C^2$
are the components of $\phi$
in the range of the projection operators
$\varPi\sb{1}=\frac 1 2(1+\beta)$,
and $\varPi\sb{2}=\frac 1 2(1-\beta)$ (under obvious
isomorphisms of these subspaces with $\C^2$).
The components \(\phi\sb 1\) (resp. \(\phi\sb 2\)) are sometimes referred to
as the electronic (resp. positronic) components, although
strictly speaking this terminology should only be used
after second quantization. 
Applying $\varPi\sb{1}$ and $\varPi\sb{2}$
to \eqref{omega-phi-is},
we have:
\begin{equation}\label{mds1}
\omega{\phi\sb{1}}
={\bm\sigma\cdot}(-i\bm\nabla-{\ech}\mathbf{A}){\phi\sb{2}}+m{\phi\sb{1}}+{\ech}A\sp 0{\phi\sb{1}},
\end{equation}
\begin{equation}\label{mds2}
\omega\phi\sb 2
={\bm\sigma\cdot}(-i\bm\nabla-{\ech}\mathbf{A}){\phi\sb{1}}-m{\phi\sb{2}}+{\ech}A\sp 0{\phi\sb{2}},
\end{equation}
\begin{equation}\label{mds3}
-\Delta A\sp 0=\ech(\phi_1\sp\ast\phi_1+\phi_2\sp\ast\phi_2),
\qquad
-\Delta\mathbf{A}
=\ech\phi\sp\ast\bm\alpha\phi
=\ech\big(
\phi\sb{1}\sp\ast\bm\sigma{\phi\sb{2}}
+\phi\sb{2}\sp\ast\bm\sigma{\phi\sb{1}}\big).
\end{equation}
We write \eqref{mds3} as
\begin{equation}
A\sp 0=\ech\bfN*\,(\phi_1\sp\ast\phi_1+\phi_2\sp\ast\phi_2)\,,
\qquad
\mathbf{A}
=\ech\bfN*\,(\phi\sb{1}\sp\ast\bm\sigma{\phi\sb{2}}
+\phi\sb{2}\sp\ast\bm\sigma{\phi\sb{1}})\,,
\label{nl2}
\end{equation}
and regard the potentials $A\sp 0$ and $\mathbf{A}=\{A\sp j\}_{j=1}^{3}$
as non-local functionals of $\phi=\begin{bmatrix}\phi_1 \\\phi_2\end{bmatrix}$.
Above,
\begin{eqnarray}\label{def-Newton}
\bfN(x)=\frac{1}{4\pi|x|},
\qquad
x\in\R^3\setminus\{0\},
\end{eqnarray}
is the Newtonian potential.
In abstract terms, the equations are of the form
$\omega{{Q}}'=\mathcal{E}'$ where
the charge functional is
\begin{eqnarray}\label{def-q}
{{Q}}(\phi)=\int\sb{\R^3}\phi\sp\ast(x)\phi(x)\,dx
\end{eqnarray}
(cf. \eqref{def-q0}),
and, regarding $A\sp 0,\,\mathbf{A}$
as non-local functionals \eqref{nl2}
of $\phi$,
the Hamiltonian
$\mathcal{E}(\phi)$ is given by
\begin{equation}\label{def-e}
\mathcal{E}(\phi)=\int\,
\Bigl(
-i\phi\sp\ast{\bm\alpha\cdot}\bm\nabla\phi
+m\phi\sp\ast\beta\phi
+\frac{{\ech}}{2}
\big(
  A\sp 0\phi\sp\ast\phi-\mathbf{A}\cdot(\phi\sp\ast\bm\alpha\phi)
\big)
\,\Bigr)\,dx.
\end{equation}
We record the following formulae for the functional derivatives:
\[
\frac{\delta{{Q}}}{\delta\phi(x)}\,=\,\phi\sp\ast(x)\,,\qquad
\frac{\delta\mathcal{E}}{\delta\phi(x)}\,=\,\left(
\bm\alpha\cdot(-i\bm\nabla-\ech \mathbf{A})\phi
+m\beta\phi
+\ech A\sp 0\phi
\right)\sp\ast(x)\,;
\]
\[
\frac{\delta{{Q}}}{\delta\phi\sp\ast(x)}\,=\,\phi(x)\,,\qquad
\frac{\delta\mathcal{E}}{\delta\phi\sp\ast(x)}\,=\,\left(
\bm\alpha\cdot(-i\bm\nabla-\ech \mathbf{A})\phi
+m\beta\phi
+\ech A\sp 0\phi
\right)(x)\,.
\]
If, say, \(\mathcal{E} \) has a directional derivative at \(\phi\in H^1(\R^3;\C^4) \)
along the direction \(f\in \mathscr{S}(\R^3;\C^4) \),
then\footnote{Recall that \(\ast \) is Hermitian conjugate, so for example \(f\sp\ast \)
  and \(\frac{\delta\mathcal{E}}{\delta\phi\sp\ast(x)}\) are, respectively, row and column vectors
  pointwise, so that the integrand is a scalar.}
\begin{equation}
\frac{d}{ds}\Biggr|_{s=0}\mathcal{E}(\phi+sf)\,=\,\langle\mathcal{E}'(\phi),f\rangle\,=\,
\int\,\Bigl(\frac{\delta\mathcal{E}}{\delta\phi(x)}\,f(x)\,
+\,f\sp\ast(x)\frac{\delta\mathcal{E}}{\delta\phi\sp\ast(x)}
\Bigr)\,dx\,.
\label{last}
\end{equation}
This integral extends to define a bounded linear map on  \(L^2(\R^3;\C^4) \) which
we continue to write as \(f\mapsto \langle\mathcal{E}'(\phi),f\rangle\),
and refer to as a directional derivative.

In accordance with the heuristics in 
\S\ref{sect-heur} we introduce functions
${\varPhi\sb{1}}(y,\epsilon),\ {\varPhi\sb{2}}(y,\epsilon)\in\C^2$
and
$\eurA\sp\mu(y,\epsilon)$
by the following scaling relations:
\begin{eqnarray}\label{ansatz-2}
\begin{array}{l}
{\phi\sb{1}}(x,\omega)=\epsilon^3{\varPhi\sb{1}}(\epsilon x,\epsilon),
\qquad
{\phi\sb{2}}(x,\omega)=\epsilon^2{\varPhi\sb{2}}(\epsilon x,\epsilon),
\\[1.5ex]
A\sp 0(x,\omega)=\epsilon^2\eurA\sp 0(\epsilon x,\epsilon),
\qquad
A\sp j(x,\omega)=\epsilon^3\eurA\sp j(\epsilon x,\epsilon),
\end{array}
\end{eqnarray}
where
$\epsilon\in(0,m)$ and $\omega\in(-m,0)$ are related by
$\omega=-\sqrt{m^2-\epsilon^2}$.
Then, writing $\bm\nabla\sb{y}$ for the
gradient with respect to $y^j=\epsilon x^j$,
$1\le j\le 3$,
the system \eqref{mds1}--\eqref{mds3}
can be written as follows:
\begin{eqnarray}
\label{sys-phi1}
&&
-2m{\varPhi\sb{1}}
+i\bm\sigma{\cdot\bm\nabla}\sb{y}{\varPhi\sb{2}}
-\epsilon^2{\ech}\eurA\sp 0{\varPhi\sb{1}}
=
-(m+\omega){\varPhi\sb{1}}
-\epsilon^2{\ech}{\eubA\cdot}\bm\sigma{\varPhi\sb{2}},
\\
\label{sys-phi2}
&&
\frac{1}{2m}{\varPhi\sb{2}}
+i\bm\sigma{\cdot\bm\nabla}\sb{y}{\varPhi\sb{1}}-{\ech}\eurA\sp 0{\varPhi\sb{2}}
=
\Big(\frac{1}{2m}-\frac{1}{m-\omega}\Big){\varPhi\sb{2}}
-\epsilon^2{\ech}{\eubA\cdot}\bm\sigma{\varPhi\sb{1}},
\\
\label{sys-a}
&&
\eurA\sp 0
={\ech}\bfN\ast\bigl(\varPhi\sb{2}\sp\ast{\varPhi\sb{2}}
+\epsilon^2\varPhi\sb{1}\sp\ast{\varPhi\sb{1}}\bigr)\,,
\qquad
\eubA
={\ech}\bfN\ast\bigl(\varPhi\sb{1}\sp\ast\bm\sigma{\varPhi\sb{2}}
+\varPhi\sb{2}\sp\ast\bm\sigma{\varPhi\sb{1}}\bigr)\,.
\qquad
\quad
\end{eqnarray}
Recall that $\varphi_0\in\mathscr{S}(\R^3)$
is the ground state solution to
the stationary Choquard equation
\eqref{lambda-phi}
with $\omega\sb 0=-\frac{1}{2m}$:
\begin{equation}\label{NR1}
\frac{1}{2m}\varphi_0
-\frac{1}{2m}\Delta \varphi_0
-
{\ech}^2\big(\bfN\ast\varphi_0^2\big)\varphi_0=0.
\end{equation}
That is,
$\varphi_0(y)$ is a strictly positive, spherically symmetric,
smooth,
strictly monotonically decaying (as a function of $\abs{y}$)
function of Schwartz class.
As discussed in the previous section,
such a solution exists by \cite{MR0471785}.
Using $\varphi_0$, we can
produce a solution
to \eqref{sys-phi1}--\eqref{sys-a}
in the nonrelativistic limit $\epsilon=0$:
\begin{eqnarray}\label{NR2}
\hat{\varPhi}(y)=\begin{bmatrix}\hat\varPhi\sb 1(y) \\
\hat\varPhi\sb 2(y)\end{bmatrix}\in\C^4,
\end{eqnarray}
with
$\ \hat\varPhi\sb 2(y)
=\varphi_0(y)\begin{bmatrix}1\\0\end{bmatrix}\ $
and
$\ \hat\varPhi\sb 1(y)
=\frac{i}{2m}\bm\sigma{\cdot\bm\nabla}\sb{y}\hat\varPhi\sb 2(y)$;
\begin{eqnarray}\label{NR3}
&&
\hat\eurA\sp 0(y)=\ech\bfN*\varphi_0^2\,,\qquad
\hat\eurA\sp 1(y)=-\frac{\ech}{m}\bfN*\varphi_0\partial_2\varphi_0\,,\qquad
\nonumber
\\[1ex]
&&
\hat\eurA\sp 2(y)=+\frac{\ech}{m}\bfN*\varphi_0\partial_1\varphi_0\,,\qquad
\hat\eurA\sp 3(y)=0.
\end{eqnarray}
The symmetry of this configuration
is axial, with the magnetic field along the $z$ axis of symmetry.

In order to describe the maps \(\varPhi\mapsto\eurA\sp\mu\) precisely, we
recall
(see e.g. \cite{MR894477})
that a homogeneous polynomial of degree \(n\)
which maps \(\varPhi\in E\) to  \(\P(\varPhi)\in F\),
from a Banach space \(E\) to a Banach space \(F\), is a mapping
of the form \(\P(\varPhi)=\A(\varPhi,\dots,\varPhi)\) where
\(\A\) is a bounded \(n\)-linear symmetric map \(E\times\dots\times E\to F\).
A polynomial is a finite sum of such homogeneous polynomials,
and an analytic mapping
\(E\to F\) is one given locally as an absolutely and uniformly
convergent power series of
polynomials. Such mappings are automatically smooth.

\begin{lemma}\label{lemma-lum}
\begin{enumerate}
\item
\label{lemma-lum-1}
Let $\varPhi=\begin{bmatrix}\varPhi_1 \\\varPhi_2\end{bmatrix}
\in H\sp 1(\R^3,\C^4)$.
Then $\eurA\sp\mu$ defined by \eqref{sys-a}
satisfy
\(
\eurA\sp\mu\in L\sp\infty(\R^3)\),
\(0\le \mu\le 3\). Furthermore the mappings
\(\varPhi
\mapsto \eubA\)
are degree 2 polynomial
mappings
\(H^1(\R^3,\C^4)\to L^\infty(\R^3)\), and
similarly
\((\varPhi,\epsilon)
\mapsto \eurA\sp 0\) is a polynomial mapping
\(H^1(\R^3,\C^4)\times \R\to L^\infty(\R^3)\).
\item
\label{lemma-lum-2}
The formulae \eqref{sys-a} also define mappings
\[
H^1(\R^3,\C^4)\to \dot H^1(\R^3),
\qquad\qquad
\varPhi
\mapsto \eubA
\]
and
\[
\quad
H^1(\R^3,\C^4)\times \R\to \dot H^1(\R^3),
\qquad
(\varPhi,\epsilon)
\mapsto \eurA\sp 0
\]
which are polynomial mappings into the homogeneous Dirichlet space
\(\dot H^1(\R^3)\).
\item
\label{lemma-lum-3}
Let $\varPhi=\begin{bmatrix}\varPhi_1 \\\varPhi_2\end{bmatrix}
\in H\sp 2(\R^3,\C^4)$.
Differentiation of \eqref{sys-a} gives mappings
\(\varPhi
\mapsto \bm\nabla\eubA\) and
\((\varPhi,\epsilon)
\mapsto \bm\nabla\eurA\sp 0\)
which are polynomial mappings
\[
H^2(\R^3,\C^4)\to L^{\infty}(\R^3)
\qquad
\mbox{and}
\qquad
H^2(\R^3,\C^4)\times \R\to L^{\infty}(\R^3),
\]
respectively.
\end{enumerate}
\end{lemma}

\begin{proof}
{\it(\ref{lemma-lum-1})}
The functions $\eurA\sp\mu$
defined by \eqref{sys-a}
are of the form $\bfN* h$ with 
$h:=f g$,
where $f,\,g\in H^1(\R^3)$.
Due to the Sobolev embedding
$H\sp 1(\R^3)\subset L\sp 6(\R^3)$,
the mapping $(f,\,g)\mapsto h=f g$ is a continuous
bilinear map $H^1(\R^3)\times H^1(\R^3)\to L\sp 1(\R^3)\cap L\sp 3(\R^3)$.
Also \(\bfN*h=(\chi\sb{\mathbb{B}^3_1}\bfN)*h+((1-\chi\sb{\mathbb{B}^3_1})\bfN)*h\)
where
$\mathbb{B}^3_1$
is the unit ball in $\R^3$
and $\chi\sb{\mathbb{B}^3_1}$
is its characteristic function.
It follows from the H\"older inequality that
\[
\|\bfN\ast h\|_{L\sp\infty}\,\leq\,
\norm{(\chi\sb{\mathbb{B}^3_1}\bfN)}_{L^{\frac{3}{2}}}\norm{h}\sb{L^3}+
\norm{((1-\chi\sb{\mathbb{B}^3_1})\bfN)}_{L^\infty}\norm{h}\sb{L^1}\,,
\]
so that the mapping \(L^1\cap L^3\ni h\mapsto \bfN\ast h\in L^\infty\)
is a continuous linear map. It follows that the composition
\((f,g)\mapsto \bfN\ast (f g)\) is a polynomial mapping
\(H^1\times H^1\to L^\infty\).

To prove {\it(\ref{lemma-lum-2})},
we recall that by the Riesz representation theorem
the linear operator \((-\Delta)^{-1}=\bfN*\) is bounded 
\(L^{6/5}(\R^3)\to\dot H^1(\R^3)\) since \(L^{6/5}=(L^6)'\subset (\dot H^1)'\).
The result therefore follows from the fact that (continuing with the same 
notation) the mapping $(f,\,g)\mapsto h=f g$ is a continuous
bilinear map
\[
H^1(\R^3)\times H^1(\R^3)\to L\sp 1(\R^3)\cap L\sp 3(\R^3)
\subset L^{6/5}(\R^3).
\]

The statement
{\it(\ref{lemma-lum-3})}
is proved by noting that a similar structure holds for the
differentiated versions of formulae \eqref{sys-a} by the Leibniz rule, 
and so the same proof works. \end{proof}




Let
\begin{eqnarray}\label{def-xy}
X=H^2(\R^3;\C^2)\oplus H^2(\R^3;\C^2)\,,
\qquad
Y=H^1(\R^3;\C^2)\oplus H^1(\R^3;\C^2)\,,
\end{eqnarray}
and define the corresponding exponentially weighted spaces,
using the norms
introduced in Corollary \ref{bounds}:
\begin{eqnarray}
\begin{cases}
X^\theta=H^{2,\theta}(\R^3;\C^2)\oplus H^{2,\theta}(\R^3;\C^2)
\,,
\\
Y^\theta=H^{1,\theta}(\R^3;\C^2)\oplus H^{1,\theta}(\R^3;\C^2)\,,
\end{cases}
\qquad
\theta\ge 0.
\label{def-xytheta}
\end{eqnarray}
The case \(\theta=0 \) reduces to the standard Sobolev norms.

Introducing the notation
\begin{equation}\label{def-not-p}
{\mom}=-i\bm\sigma{\cdot\bm\nabla}\sb{y},
\end{equation}
we rewrite
\eqref{sys-phi1}, \eqref{sys-phi2}
as the equation $\mathcal{F}=0$,
where (for small nonnegative \(\theta \))
\begin{eqnarray}\label{df}
&&
\mathcal{F}:\;
X^\theta\times(-m,+m)\longrightarrow Y^\theta,
\\[2ex]
\nonumber
&&
\mathcal{F}:\;
(\varPhi,\,\epsilon)
\mapsto
\begin{bmatrix}
\displaystyle
2m{\varPhi\sb{1}}
+{\mom}
{\varPhi\sb{2}}+\epsilon^2{\ech}\eurA\sp 0{\varPhi\sb{1}}
-(m+\omega){\varPhi\sb{1}}
-\epsilon^2{\ech}{\eubA\cdot}\bm\sigma{\varPhi\sb{2}}
\\[1ex]
\displaystyle
-\frac{1}{2m}{\varPhi\sb{2}}
+{\mom}{\varPhi\sb{1}}+{\ech}\eurA\sp 0{\varPhi\sb{2}}
+
\Big(
\frac{1}{2m}-\frac{1}{m-\omega}
\Big)
{\varPhi\sb{2}}
-\epsilon^2{\ech}{\eubA\cdot}\bm\sigma{\varPhi\sb{1}}\,
\end{bmatrix}.
\end{eqnarray}
Above, $\omega=-\sqrt{m^2-\epsilon^2}$.
As before, we regard the 
$\eurA\sp \mu=(\eurA\sp 0,\eubA)$,
$\ \eubA=\{\eurA\sp j\}_{j=1}^3$,
as non-local functionals
$\eurA\sp \mu=\eurA\sp \mu(\varPhi,\epsilon)$
determined by \eqref{sys-a}.
With this understood,
the entire system \eqref{mds1}--\eqref{mds3} is encapsulated
in the equation $\mathcal{F}(\varPhi,\epsilon)=0$ for
$\varPhi=\begin{bmatrix}\varPhi_1 \\\varPhi_2\end{bmatrix}$ only.
We note that in terms of the functionals
${{Q}}$ and $\mathcal{E}$
defined by \eqref{def-q}, \eqref{def-e},
one has
\[
\mathcal{F}(\varPhi,\epsilon)\,=\,
\begin{bmatrix}
\epsilon^{-3} & 0                 \\
0& \epsilon^{-4}
\end{bmatrix}
(\mathcal{E}'-\omega{{Q}}')
\left(
\begin{bmatrix}
\epsilon^3{\varPhi\sb{1}}
\\
\epsilon^2{\varPhi\sb{2}}
\end{bmatrix}
\right).
\]
The nonrelativistic limit \(\hat\varPhi\) satisfies 
$\mathcal{F}(\hat\varPhi,0)=0$
(cf.~\eqref{NR2}, \eqref{NR3}),
so that
to obtain solutions for small $\epsilon$ by the implicit function theorem
it is necessary to
compute the derivative of $\mathcal{F}(\varPhi,\epsilon)$ at the point
$(\hat\varPhi,0)$. This is 
determined by the set of directional derivatives.
Define
$\bme_1=\begin{bmatrix}1            \\0\end{bmatrix}$ and
$\bme_2=\begin{bmatrix}0            \\1\end{bmatrix}$, and
let
$g\in H^1(\R^3,\C^2)$.
To compute the directional derivatives, first note that
$\eurA\sp j$ drops out on putting $\epsilon=0$, and then
note further that by \eqref{sys-a} only the derivative 
of $\eurA\sp 0$  at 
$(\Phi,\epsilon)=(\hat\varPhi,0)$ with respect to
$\varPhi\sb 2$ is nonzero, with derivative given by
$$
\frac{d}{dt}\eurA\sp 0
\Big(\begin{bmatrix}\hat\varPhi_1 \\\hat\varPhi_2+t g\end{bmatrix},\epsilon
\Big)
\bigg|_{t=0,\epsilon=0}
=2{\ech}\bfN\ast\bigl(\varphi_0\Re\langle \bme_1,g\rangle\sb{\C^2}\bigr),
$$
with the Newtonian potential from \eqref{def-Newton},
where
$
\langle v,w\rangle\sb{\C^2}=\bar v_1 w_1+\bar v_2 w_2
$
is the complex sesquilinear inner product of $v,\,w\in\C^2$.
We deduce that for $\C^2$-valued Schwartz functions $U$ and $V$,
\[
\frac{d}{dt}\mathcal{F}
\Big(
\begin{bmatrix}
\hat\varPhi_1+t U                 \\
\hat\varPhi_2+t V\end{bmatrix}
,\epsilon
\Big)
\big|_{t=0,\epsilon=0}
=\;\eurM\begin{bmatrix}U          \\V\end{bmatrix}
\,,
\]
where
\begin{equation}\label{def-m}
\eurM=
\begin{bmatrix}
2m&{\mom}
\\{\mom}&-\frac{1}{2m}+{\ech}\hat\eurA\sp 0+2
{\ech}^2
\varphi_0 \bme_1
\bfN\ast(\varphi_0\Re\langle \bme_1,\,\cdot\,\rangle\sb{\C^2})
\end{bmatrix}
\end{equation}
and ${\mom}=-i\bm\sigma{\cdot\bm\nabla}\sb{y}$
was introduced in \eqref{def-not-p}.
Thus the derivative of $\mathcal{F}$ at
the nonrelativistic limit point $(\hat\varPhi,0)$
is the linear map
$D\mathcal{F}(\hat\varPhi,0)$
given by the matrix
$
\eurM
$. This is a differential operator, which we consider as an
unbounded operator on $L^2(\R^3;\C^2)
\oplus L^2(\R^3;\C^2)$. 
\begin{lemma}
\label{lemma-fpc}
\begin{enumerate}
\item
\label{lemma-fpc-1}
The map
$\eurM:\,\begin{bmatrix}U         \\V\end{bmatrix}
\mapsto\begin{bmatrix}F           \\G\end{bmatrix}$
is a Hermitian operator
with domain
$X$ (cf. \eqref{def-xy}).
\item
\label{lemma-fpc-2}
For small nonnegative $\theta$,
the mapping $\eurM$
is continuous from
$X^\theta$ into $Y^\theta$
(cf.~\eqref{def-xytheta}).
\item
\label{lemma-fpc-3}
The kernel
of $\eurM$
is given by
\begin{eqnarray}
&&
\!\!\!\!\!\!
\ker\eurM
\nonumber
\\
&&
\!\!\!\!\!\!
=
\left\{
\Big(-\frac{{\mom\,} V}{2m},\,V\Big): V
=
(\mathbf{a}{\cdot\bm\nabla}\sb{y}
\varphi_0+i b\varphi_0)\bme_1+c\varphi_0\,\bme_2,
\ (\mathbf{a},\,b,\,c)\in\R^3\times\R\times\C
\right\}.
\nonumber
\end{eqnarray}
\item
\label{lemma-fpc-4}
The range of $\eurM:\,
\begin{bmatrix}U                  \\V\end{bmatrix}
\mapsto
\begin{bmatrix}F                  \\G\end{bmatrix}$
is closed in the topology of $Y$
and is given  by
\begin{eqnarray}
             &   & 
\range \eurM
=(\ker\eurM)\sp\perp
               =  
\left\{\begin{bmatrix}F           \\G\end{bmatrix}\in Y:
\,\Re \Big(\frac{{\mom\,} F}{2m} - G\Big)_1\in(\ker L_1)^\perp,
\right.
\nonumber
\\
&& 
\hspace{2.3cm}
\left.
\quad
\,\Im\Big(\frac{{\mom\,} F}{2m} - G\Big)_1
\in(\ker L_0)^\perp,
\,\Big(\frac{{\mom\,} F}{2m} - G\Big)_2\in(\ker L_0)^\perp
\right\},
\nonumber
\end{eqnarray}
where $^\perp$ is the orthogonal complement with 
respect to the inner product in $L^2\oplus L^2$.
\item
\label{lemma-fpc-5}
The inverse of
$\eurM:\,\begin{bmatrix}U         \\V\end{bmatrix}
\mapsto \begin{bmatrix}F          \\G\end{bmatrix}$
is given by
\begin{eqnarray}
\nonumber
&&
\!\!\!\!\!\!\!
U=\frac{1}{2m}\bigl({F}-{{\mom}\, V}\bigr)\,,
\\[1ex]
&&
\nonumber
\!\!\!\!\!\!\!
V=\bme_1 V_1+\bme_2 V_2
\\
&&
\nonumber
=
\left(
L_1^{-1}\Re\Big(\frac{{\mom\,} F}{2m} - G\Big)_1
+i L_0^{-1}\Im\Big(\frac{{\mom\,} F}{2m} - G\Big)_1
\right)
\bme_1
+L_0^{-1}\Big(\frac{{\mom\,} F}{2m} - G\Big)_2\bme_2\,,
\end{eqnarray}
\end{enumerate}
where
the definitions
and properties of the operators $L_0,\,L_1$ are given in \S\ref{sect-nr}
(cf. \eqref{def-l0-l1}).
\end{lemma}

\begin{proof}
The proof depends on some properties of the linearized Choquard equation
from \cite{MR2561169} which are stated 
in \S\ref{sect-nr}.
The fact in {\it(\ref{lemma-fpc-1})}
that $\eurM$ is Hermitian follows from the fact that $\mom$
is Hermitian. From Lemma~\ref{lemma-lum} the assertion
{\it(\ref{lemma-fpc-2})} is immediate from 
the properties of
$\bfN$ and the fact that $\varphi_0$ and its
partial derivatives are smooth and exponentially decreasing.
To prove
{\it(\ref{lemma-fpc-3})}, {\it(\ref{lemma-fpc-4})}, and {\it(\ref{lemma-fpc-5})},
we consider how
to solve 
$\eurM\begin{bmatrix}U \\V\end{bmatrix}=
\begin{bmatrix}F\\G\end{bmatrix}$, i.e. the system
\[
\eurM
\begin{bmatrix}
U\\V
\end{bmatrix}
=
\begin{bmatrix}
2m U+{\mom\,} V
\\[0.5ex]
{\mom\,} U
-\frac{V}{2m}
+
{\ech}\hat\eurA\sp 0 V
+2{\ech}^2\varphi_0\bme_1\bfN\ast(\varphi_0\Re V_1)
\end{bmatrix}
=\begin{bmatrix}F\\G\end{bmatrix}.
\]
We first express $U$ in terms of $V$
by
$
U=\frac{1}{2m}(F-{{\mom\,} V})\,,
$
and, writing $V=V_1\bme_1+V_2\bme_2$,
\[
\frac{{\mom\,} F}{2m}
+
\frac{\Delta V}{2m}-\frac{V}{2m}
+{\ech}\hat\eurA\sp 0 V
+2{\ech}^2\varphi_0\bme_1\bfN\ast(\varphi_0\Re V_1)
=G.
\]
Referring to the definitions of $L_0$ and $L_1$
in \S\ref{sect-nr} (cf. \eqref{def-l0-l1}), with
${\omega\sb 0}$ set equal to $-1/(2m)$,
we arrive at the following equations:
\begin{equation}
L_1 V_1=\Big(\frac{{\mom\,} F}{2m} - G\Big)_1\,,
\qquad
L_0 V_2=\Big(\frac{{\mom\,} F}{2m} - G\Big)_2\,.
\end{equation}
It is useful here that the components
with respect to $\bme_1$ and $\bme_2$  are decoupled.
The identification of the kernel in {\it(\ref{lemma-fpc-3})}
is then a specialization of this, given the information on 
$\ker L\sb 0$ and $\ker L\sb 1$ in \S\ref{sect-nr}, and also
{\it(\ref{lemma-fpc-4})}
is a consequence of the identification of the ranges of $L\sb 0$ and 
$L\sb 1$ given in \S\ref{sect-nr}
(cf. Lemma~\ref{lemma-l0-l1}).
\end{proof}

The existence statement in Theorem~\ref{theorem-sw-dm}
now almost follows from using the implicit function theorem
to solve $\mathcal{F}=0$. In order to handle the degeneracies
arising from symmetries we use the following trick from \cite{MR1708440}, which
we state as a lemma applying to functionals \(\mathcal{E} \) and \({{Q}} \)
defined on a general real Hilbert space \(H\). In the present paper the relevant
choice is \(H=L^2(\R^3;\C^4)\), with the real \(L^2\) inner product
\begin{equation}\label{rip}
  \langle\phi,\psi\rangle_{L^2}=\Re\int_{\R^3}\phi\sp\ast(x)\psi(x)\,dx\, .
  \end{equation}

\begin{lemma}\label{lemma-trick}
Let $\{{\xi}^\alpha\}_{\alpha\in I}$ be a finite 
collection of elements of a real Hilbert space \(H\),
indexed by \(I \), all lying in some subspace
\(F\subset H\) with the property that
\(\mathcal{E} \) and \({{Q}} \) are differentiable
along each direction \(f\in F \) with directional derivatives
\(\langle{{Q}}'\,,\,f\rangle \) and \(\langle\mathcal{E}'\,,\,f\rangle \)
for \(f\in F\,. \) Assume further that the $\{{\xi}^\alpha\}$
correspond to infinitesimal symmetries, in the sense that
$\langle{{Q}}'\,,\,{\xi}^\alpha\rangle=0
=\langle\mathcal{E}'\,,\,{\xi}^\alpha\rangle\,,$ for all \(\alpha\in I\,. \)
Let \(\phi \) satisfy
\begin{eqnarray}\label{omega-q-e}
\omega{{Q}}'-\mathcal{E}'-\sum_{\alpha\in I} a_\alpha {\xi}^\alpha\,\,=\,0\,,
\end{eqnarray}
for some set of numbers $a_\alpha\in\R$. Then \(a_\alpha=0 \,\forall\alpha\in I\)
as long as 
the matrix $\langle {\xi}^\alpha,{\xi}^\beta\rangle$ is nondegenerate.
\end{lemma}

\begin{proof}
Put \(f=\xi^\beta \) and make use of the assumptions, then
\(\sum_{\alpha\in I}\, a_\alpha\langle {\xi}^\alpha,{\xi}^\beta\rangle=0\),
which implies \(a_\alpha=0 \,\forall\alpha\in I\) by the nondegeneracy
of the matrix $\langle {\xi}^\alpha,{\xi}^\beta\rangle$.
\end{proof}

\begin{remark}
  \label{rrr}
It follows from the proof that
instead of \eqref{omega-q-e}
it is sufficient to assume that
\[
\Big\langle\,\omega{{Q}}'(\phi)-\mathcal{E}'(\phi)-\sum a_\alpha {\xi}^\alpha\,,\,
f\,\Big\rangle\,=\,0\,,
\qquad
\forall f\in F.
\]
 \end{remark}

\begin{example}
For a simple example consider $\psi:\R\to\C$
and ${Q}=\frac{1}{2}\int\,|\psi|^2\,dx$ and
${E}=\int\,(\frac{1}{2}|\nabla\psi|^2-\frac{1}{p+1}|\psi|^{p+1})\,dx$
the symmetry of phase rotation corresponds to the infinitesimal
symmetry ${\xi}(\psi)=i\psi$, and it is easy to check that given
an $H^1$ distributional
solution of $\omega {Q}'-{E}'-a{\xi}=0$, i.e. a
weak solution of $-\Delta\psi-|\psi|^p\psi=\omega\psi-i a\psi$,
for any $a\in\R$,
one necessarily has $a=0$. The same holds in higher dimensions
as long as $p$ is such that the equation holds as an equality in $H^{-1}$.
\end{example}

\begin{remark}\label{remark-trick}
\label{is}
The advantage of solving a more general
equation with the unknown ``multipliers'' $a_\alpha$ 
is that in an implicit function theorem setting,
the multipliers can be varied to fill out the part of
the cokernel corresponding to the symmetries. It is then shown after the
fact that the multipliers are equal to zero.
The choice of ${\xi}^\alpha$
is determined by the symmetry group;
in the case of Dirac--Maxwell the relevant group is the seven-dimensional
group generated by translations, rotations and phase rotation. Thus
the index set is \(\alpha\in\{1,\dots 7\} \) with the corresponding
multipliers \(a_\alpha\) written in order as
\( (\mathbf{a},\mathbf{b},a_0)\in\R^3\times\R^3\times\R\,. \)
The infinitesimal versions of these actions give the following vector fields
on the phase space \(H^1(\R^3;\C^4) \)
(\cite{MR0187641} or \cite[\S 3.4]{sakurai1967advanced}):
\begin{eqnarray}\label{def-xi-eta-zeta}
\bm\xi=\bm\nabla\phi\,,
\qquad
\bm\eta
=\bm{l}\phi+\frac{i}{2}
\begin{bmatrix} \bm\sigma  & 0 \\ 0 & \bm\sigma\end{bmatrix}
\phi\,,
\qquad
\zeta
=i\phi\,,
\label{angen}
\end{eqnarray}
where \(\bm{l}=\{\epsilon_{i j k} x_j\partial_k\}_{i=1}^{3}\) is the
standard angular momentum generator. The Lorentz invariance of the
Dirac construction ensures that
\begin{equation}\label{ang}
  \Biggl
      [\bm{l}\,+\,\frac{i}{2}
\begin{bmatrix} \bm\sigma  & 0 \\ 0 & \bm\sigma\end{bmatrix}\,,\,
  -i\bm\alpha\cdot\bm\nabla
  \Biggr]\,=\,0\,.
\end{equation}
For example, let \(\phi\in H^{2,\theta}(\R^3;\C^4) \) for some \(\theta>0 \);
then, since the Hamiltonian density, i.e. the integrand in \eqref{def-e}, is a scalar with
respect to Euclidean transformations, we have
\begin{align*}
\int_{\R^3 }\,&
\Bigl(
-i\tilde\phi\sp\ast{\bm\alpha\cdot}\bm\nabla\tilde\phi
+m\tilde\phi\sp\ast\beta\tilde\phi
+\frac{{\ech}}{2}
\big(
  \tilde A\sp 0\tilde\phi\sp\ast\tilde\phi-\mathbf{\tilde A}\cdot(\tilde\phi\sp\ast\bm\alpha\tilde\phi)
\big)
\,\Bigr)\,dx\\\,&=\,
\int_{\R^3 }\,
\Bigl(
-i\phi\sp\ast{\bm\alpha\cdot}\bm\nabla\phi
+m\phi\sp\ast\beta\phi
+\frac{{\ech}}{2}
\big(
  A\sp 0\phi\sp\ast\phi-\mathbf{A}\cdot(\phi\sp\ast\bm\alpha\phi)
\big)
\,\Bigr)\,dx\,,
\end{align*}
where \(\tilde\phi\,,\tilde A\sp\mu \) are obtained by the action of a spatial
rotation on \(\phi\,,A\sp\mu\,. \)
Differentiation of this integral identity with respect to the parameter of rotation
$\bm\eta=\{\eta_j\}_{j=1}^3$
and use of \eqref{ang} leads to
\begin{equation}\label{gen-is}
\langle\mathcal{E}'(\phi),\eta_j\rangle\,=\,\int_{\R^3 }\,\Bigl(\,\eta\sp\ast_j(x)\,\frac{\delta\mathcal{E}}{\delta\phi\sp\ast(x)}\,
+\,\frac{\delta\mathcal{E}}{\delta\phi(x)}\,\eta_j(x)\,\Bigr)
\,dx\,=0\,,
\qquad
1\le j\le 3\,.
  \end{equation}
The same is true for translations and phase rotations (i.e.,
the case of \(\bm\xi\) and \(\zeta \), respectively, in place of
\(\bm\eta \)). We call vector fields on the phase space such as \(\bm\eta\)
generalized infinitesimal symmetries if they are locally square integrable and
satisfy \eqref{gen-is} and \(\langle{{Q}}',\,{\bm\eta}\rangle=0\)
when \(\phi\in H^{1,\theta}(\R^3,\C^4) \) for some nonnegative
\(\theta\,. \) 
\end{remark}
\begin{example}
  As an example of Lemma~\ref{lemma-trick} for the case at hand, with \(\mathcal{E},{{Q}} \)
  as in \eqref{def-q} and \eqref{def-e}, assume that \(\phi\in H^1(\R^3;\C^4) \) is such that
\[
\omega{{Q}}'(\phi)-\mathcal{E}'(\phi)-\mathbf{a}\cdot{\bm\xi}+i a_0\zeta\,\,=\,0
\qquad(\hbox{in $L^2$}),
\]
with
$(a_0,\mathbf{a})\in\R\times\R^3$,
and \(\bm\xi \), \(\zeta \) as in \eqref{def-xi-eta-zeta}. Then in fact
\[
\omega{{Q}}'(\phi)-\mathcal{E}'(\phi)\,=\,0
\qquad(\hbox{in $L^2$})\,.
\]
\end{example}

  In order to treat the rotational symmetry \(\bm\eta \) a technical modification is needed
  on account of the linear growth at infinity of the coefficient in the angular
  momentum vector field \(\bm{l}=\{\epsilon_{i j k}x_j\partial_k\}_{i =1}^{3} \),
  which potentially means that
  \(\bm\eta \) might not be square integrable. The most efficient way to circumvent this
  issue seems to be to work in the exponentially weighted spaces \(H^{s,\theta} \) defined above.
  The following lemma, which is proved in exactly the same way as
  Lemma~\ref{lemma-trick}, gives a slightly more general setting than needed.
\begin{lemma}\label{lemma-trick-gen}
  Let \(\phi\in H^{2,\theta}(\R^3;\C^4) \) for some \(\theta\geq 0 \), and assume
  there is a finite set  $\{{\xi}^\alpha\}_{\alpha\in I}$ of generalized infinitesimal
  symmetries, in the sense of Remark~\ref{remark-trick}, which
  all lie in some subspace \(F\subset L^2_{\mathrm{loc}}(\R^3;\C^4) \).
  Assume that \(\phi\) satisfies
\[
\omega{{Q}}'(\phi)\,-\,\mathcal{E}'(\phi)\,
  -\,\sum_{\alpha\in I} a_\alpha {\xi'}^\alpha(x)\,=0\
\] 
for some set of numbers $a_\alpha\in\R$,
and for some finite set $\{{\xi'}^\alpha\}_{\alpha\in I}$
of elements of \(F' \), the dual space of
\(F \). If the matrix with entries $\langle {\xi'}^\alpha,{\xi}^\beta\rangle_{L^2}$,
computed using the inner product \eqref{rip},
 is nondegenerate
then \(a_\alpha=0 \ \ \forall\alpha\in I\).
\end{lemma}
In the case at hand, under the assumption
\(\phi\in H^{2,\theta}(\R^3;\C^4) \) for some \(\theta>0 \),
all the vector fields in \eqref{def-xi-eta-zeta} are actually square
integrable, but it is nevertheless necessary to introduce a spatial cut-off
into the  definition of the \(\xi'_\alpha \) for which
the nondegeneracy assumption holds,
see below. We are looking for $\varPhi(\epsilon)$ in the form
\begin{equation}\label{varphi-varpsi}
\varPhi(\epsilon)=\hat\varPhi+\varPsi(\epsilon),
\qquad
\varPsi(0)=0.
\end{equation}
We use the same  component notation as above:
$
\hat\varPhi=\begin{bmatrix}\hat\varPhi\sb 1
\\\hat\varPhi\sb 2\end{bmatrix}\in\C^4\,,
$
$
\varPsi=\begin{bmatrix}\varPsi\sb 1
\\\varPsi\sb 2\end{bmatrix}\in\C^4\,.
$
To make use of Lemma~\ref{lemma-trick-gen} we will apply the 
implicit function theorem to the function
\begin{eqnarray}\label{def-g}
\mathcal{G}_R(\varPsi,\mathbf{a},\mathbf{b},\epsilon)
&=&
\mathcal{F}(\hat\varPhi+\varPsi,\epsilon)\,
\,+\,
\chi_R\mathbf{a}{\cdot\bm\nabla}\sb{y}
\begin{bmatrix}\epsilon(\hat\varPhi\sb 1+\varPsi\sb 1) \\
\hat\varPhi\sb 2+\varPsi\sb 2\end{bmatrix}
\\
\nonumber
&&+
\chi_R
\,\mathbf{b}\cdot
\left(
\begin{bmatrix}
\epsilon\bm{l}(\hat\varPhi\sb 1+\varPsi\sb 1)
\\
\bm{l}(\hat\varPhi\sb 2+\varPsi\sb 2)
\end{bmatrix}
+
\frac{i}{2}
\begin{bmatrix}
\epsilon\bm\sigma(\hat\varPhi\sb 1+\varPsi\sb 1)
\\
\bm\sigma(\hat\varPhi\sb 2+\varPsi\sb 2)
\end{bmatrix}
\right)\,.
\end{eqnarray}
Here \(\chi_R(\cdot)=\chi(\cdot/R) \), where
$R\ge 1$ and \(\chi\in C_0^\infty(\R^3) \) is a radially
symmetric function which satisfies
\(\chi(y)=1 \) for \(|y|\leq 1 \) and 
\(\chi(y)=0 \) for \(|y|> 2 \).
\begin{remark}
Referring to Remark~\ref{is}, we have introduced a linear combination
of the six infinitesimal symmetries corresponding to translation and
rotation, but with a spatial cut-off enforced by multiplication
by \(\chi_R \),
replacing \(\bm\xi,\bm\eta \) by
\[
{\bm\xi}_R=\chi_R\bm\nabla\phi\,,\quad
{\bm\eta}_R
=\chi_R\biggl(\bm{l}\,+\frac{i}{2}
\begin{bmatrix} \bm\sigma & 0\\ 0 & \bm\sigma\end{bmatrix}
\biggr)
\phi\,,
\qquad
R\ge 1,
\] 
respectively. (It is not necessary to also introduce a multiplier for
phase rotation due to the presence of infinitesimal rotation around $x_3$-axis).
In terms of the original variables
(cf. \eqref{def-xi-eta-zeta}):
\begin{equation}\label{g-r-a-b}
\mathcal{G}_R(\varPsi,\mathbf{a},\mathbf{b},\epsilon)\,=\,
\begin{bmatrix}
\epsilon^{-3}             & 0                          \\
0                         & \epsilon^{-4}
\end{bmatrix}
\Bigl(\mathcal{E}'-\omega{{Q}}'
+\epsilon \mathbf{a}\cdot\bm\xi_R+\epsilon^2 \mathbf{b}\cdot\bm\eta_R\Bigr),
\end{equation}
evaluated at
$\phi=\begin{bmatrix}\epsilon^3(\hat\varPhi\sb 1+\varPsi\sb 1)
                                                       \\
\epsilon^2 (\hat\varPhi\sb 2+\varPsi\sb 2)
\end{bmatrix}
$
\,.

\end{remark}

The idea is to solve \(\mathcal{G}_R=0 \) for some fixed 
large \(R\gg 1 \) and then to show
that this actually gives solutions to \(\mathcal{F}=0 \) for \(\epsilon \)
sufficiently small. The spatial cut-off ensures that \(\mathcal{G}_R \)
is well-behaved on the Sobolev spaces \(H^{s,\theta}\,. \)

\begin{proposition}\label{prop-sw-dm}
There is ${\epsilon\sb\ast}>0$ such that
for $\epsilon\in(-\epsilon\sb\ast,{\epsilon\sb\ast})$
there is a solution to \eqref{mds1}--\eqref{mds3},
with $\omega=-\sqrt{m^2-\epsilon^2}$,
given by the Ansatz \eqref{ansatz-2} with 
$\varPhi(\epsilon)=\hat{\varPhi}+\varPsi(\epsilon)$
obtained from a $C^\infty$-function
\[
\varPsi\in C^\infty\bigl(
(-{\epsilon\sb\ast},{\epsilon\sb\ast})\,;\,
H^{2,\theta}(\R^3;\C^4)
\cap \ker\eurM^\perp
\bigr)
\]
for small positive \(\theta \), and
satisfying $\varPsi(0)=0$,
and
$\eurA\sp 0\in C^\infty\bigl(
(-{\epsilon\sb\ast},{\epsilon\sb\ast})\,;\, \dot H\sp 1\cap L^\infty\bigr)$,
$\eurA\sp j\in C^\infty\bigl(
(-{\epsilon\sb\ast},{\epsilon\sb\ast})\,;\,
\dot H\sp 1\cap L^\infty\bigr)
$
given by \eqref{sys-a}:
\begin{eqnarray*}
 &  & 
\eurA\sp 0={\ech}\bfN*\,(\varPhi_1\sp\ast\varPhi_1+
\epsilon^2\varPhi_2\sp\ast\varPhi_2)
\in C^\infty\bigl(
(-{\epsilon\sb\ast},{\epsilon\sb\ast})\,;\, \dot H\sp 1\cap L^\infty\bigr)
\,,
 \\
 &  & 
\eubA={\ech}\bfN*\,(\varPhi\sb{1}\sp\ast\bm\sigma{\varPhi\sb{2}}
+\varPhi\sb{2}\sp\ast\bm\sigma{\varPhi\sb{1}})
\in C^\infty\bigl(
(-{\epsilon\sb\ast},{\epsilon\sb\ast})\,;\,
\dot H\sp 1\cap L^\infty\bigr)
\,.
\end{eqnarray*}
Above,
$\dot H^1=\dot H^1(\R^3,\R)$ is the homogeneous Dirichlet
space of $L^6$ functions with
\[
\|f\|_{\dot H^1}^2:=\int\sb{\R^3}|\nabla f|^2\,dx<\infty.
\]
One has
\begin{equation}\label{phi1-phi2-epsilon-2}
\norm{\varPhi(\epsilon)-\hat\varPhi}\sb{H\sp 2}
=O(\epsilon^2)\,,
\qquad
\epsilon\in(-\epsilon\sb\ast,\epsilon\sb\ast)\,.
\end{equation}
The functions
$\varPhi_1(y,\epsilon)$, $\varPhi_2(y,\epsilon)$
are even in $\epsilon$.
\end{proposition}

\begin{proof}
The proof of existence of solutions to 
\eqref{mds1}--\eqref{mds3}
is by the implicit function theorem and Lemma~\ref{lemma-trick},
perturbing from the nonrelativistic limit point
$\mathcal{F}(\hat\varPhi,0)=0$.
To start, we claim that $\mathcal{F}$, as defined in 
\eqref{df}, is a $C^\infty$-function
\[
\mathcal{F}:\,X^\theta\times (-m,+m)\to Y^\theta, \qquad\theta\ge 0.
\]
To prove this, we notice that the expression for $\mathcal{F}$
is manifestly smooth in $\epsilon$ for $\epsilon^2<m^2,$ and
its dependence on $\varPhi_j$ is built up from compositions of 
certain multilinear maps and linear operators;
the structure of the expressions obtained after
successive differentiation is the same. Referring to the specific formulae, 
the fact that these expressions are all $C^\infty$
is an immediate consequence of  Lemma~\ref{lemma-lum} and the fact that multiplication
gives continuous bilinear (and hence smooth) maps $H^{1,\theta}\times H^{2,\theta}
\to H^{1,\theta}$ and
$H^{2,\theta}\times H^{2,\theta}\to H^{2,\theta}$
(Moser inequalities) for \(\theta\geq 0\,. \) For example, consider the term
\begin{equation}\label{tmap}
\eurA\sp 0{\varPhi\sb{2}}\,=\,
\ech
\bfN\ast\Bigl(\varPhi\sb{2}\sp\ast{\varPhi\sb{2}}
+\epsilon^2\varPhi\sb{1}\sp\ast{\varPhi\sb{1}}\Bigr)\,
{\varPhi\sb{2}}\,,
\end{equation}
for $\varPhi_1,\,\varPhi_2\in H\sp 2(\R^3,\C^2)$.
By Lemma~\ref{lemma-lum}, both \(\eurA\sp 0\) and \(\bm\nabla\eurA\sp 0\)
are bounded in \(L^\infty\), and consequently since \(\varPhi_2\in H^{2,\theta}\),
the product rule implies that \(\eurA\sp 0\varPhi_2\) is bounded
in \(H^{1,\theta}\). On the other hand, the mapping
\eqref{tmap} is cubic and can be expressed in an obvious way
as a composition of the 
embedding
\begin{eqnarray}
\nonumber
H^{2,\theta}\times (-m,+m)
&\to&
H^{2,\theta}\times H^{2,\theta}\times H^{2,\theta}\times(-m,+m)\times (-m,+m),
\\
\nonumber
(\varPhi,\epsilon)
&\mapsto&
(\varPhi,\varPhi,\varPhi,\epsilon,\epsilon)
\end{eqnarray}
with a
mapping into \(H^{1,\theta}\) which is both multilinear and bounded (by identical reasoning
to that in the previous sentence).
The composition is therefore
smooth by the chain rule. Analogous reasoning for the other terms
shows that \(\mathcal{F}\) defines a smooth mapping
$X^\theta\times (-m,+m)\to Y^\theta$ as required.

Computing the derivatives of \eqref{def-g} at $\epsilon=0$, $\varPsi=0$
and using the spherical symmetry of the ground state solution of \eqref{NR1}, 
we see that the functions
$\langle
\{
\p\sb{\mathbf{a}\sb j}{\mathcal{G}_R},
\p\sb{\mathbf{b}\sb j}{\mathcal{G}_R}
\sothat
1\le j\le 3
\}\rangle$
converge strongly in \(L^2(dy) \) as \(R\to+\infty \) to the basis for
$\ker\eurM$ given in Lemma~\ref{lemma-fpc}.
This establishes that if \(R \) is sufficiently large (depending
only on \(\varphi_0 \)), then the derivative of ${\mathcal{G}_R}$ at
$\epsilon=0,\varPsi=0$,
$\mathbf{a}=0$, $\mathbf{b}=0$
with respect to 
$(\Psi,\mathbf{a},\mathbf{b})$ is a linear homeomorphism from 
$\big((\ker\eurM)^\perp\cap X^\theta\big)\times \R^3\times\R^3$ {\em onto} $Y^\theta$
for small positive \(\theta \).

It follows that for such \(R \) there is $\epsilon\sb\ast>0$
such that
there exist $C^\infty$-functions
$\epsilon\mapsto(\varPsi(\epsilon),\mathbf{a}(\epsilon),\mathbf{b}(\epsilon))
\in X\times \R^3\times\R^3$,
defined for $\epsilon\in(-\epsilon\sb\ast,\epsilon\sb\ast)$,
such that
\begin{equation}\label{varphi-perp-ker-m}
{\mathcal{G}_R}
\big(
\varPsi(\epsilon),\mathbf{a}(\epsilon),\mathbf{b}(\epsilon),\epsilon
\big)=0,
\qquad
\varPsi(\epsilon)
\perp
\ker\eurM,
\qquad
\epsilon\in(-\epsilon\sb\ast,\epsilon\sb\ast).
\end{equation}
(This latter condition serves to divide out
by the action of the symmetry group, giving a local slice.)
Referring to Lemma~\ref{lemma-trick-gen}, to deduce that these in
fact generate solutions of $\mathcal{F}=0$
for sufficiently small $\epsilon>0$,
it is sufficient to verify that
$\mathbf{a}(\epsilon)=0$,
$\mathbf{b}(\epsilon)=0$,
which
is in turn a consequence of the nondegeneracy of the appropriate matrix of
inner products, scaled as above. This amounts
to the need to verify nondegeneracy of the $6\times 6$ matrix
\begin{equation}\label{xi-xi}
\begin{bmatrix}
\langle\p\sb{y\sp j}{}\phi,\chi_R\,\p\sb{y\sp k}{}\phi\rangle
\phantom{\int\sb\int}
 & 
\langle\p\sb{y\sp j}{}\phi,\chi_R\,(l_{k'}+\frac i 2\Sigma\sb{k'})\phi\rangle
 \\
\langle(l_{j'}+\frac i 2\Sigma\sb{j'})\phi,
\chi_R\,\p\sb{y\sp k}\phi\rangle
 & 
\langle
(l_{j'}+\frac i 2\Sigma\sb{j'})\phi
,\chi_R(l_{k'}+\frac i 2\Sigma\sb{k'})
\phi\rangle
\end{bmatrix}
\end{equation}
for small $\epsilon$.
(In the matrix \eqref{xi-xi} the indices $j$, $j'$, $k$, $k'$ run 
between $1$ and $3$.)

\begin{lemma}\label{lemma-xi-xi}
For fixed \( R  \) chosen sufficiently large, the matrix
given by \eqref{xi-xi},
evaluated at
$\phi(x,\omega)=\begin{bmatrix}\epsilon^3(\hat\varPhi\sb 1+\varPsi\sb 1)
 \\
\epsilon^2 (\hat\varPhi\sb 2+\varPsi\sb 2)
\end{bmatrix}\biggr |_{y=\epsilon x}
$,
\ $\omega=-\sqrt{m^2-\epsilon^2}$,
is nondegenerate for sufficiently small $\epsilon>0$.
\end{lemma}

\begin{proof}
Clearly
the dominant terms arise from the second (``large'') component
$\epsilon^2 (\hat\varPhi\sb 2+\varPsi\sb 2)$,
giving rise to diagonal matrix elements which, referring to
the block form in \eqref{xi-xi}, are
$O(\epsilon^4)$.
Using $\|\varPsi_j\|_{H^2}=O(\epsilon)$,
we will deduce the result from 
nondegeneracy of the matrix with $\varPsi\sb j$
set equal to zero and \(R=+\infty\,. \)
To start with, using
\[\epsilon^{-2}\phi
=
\begin{bmatrix}-\frac{\epsilon}{2m}{\mom\,}
\hat\varPhi\sb 2
                                          \\
\hat\varPhi\sb 2
\end{bmatrix}
+
\begin{bmatrix}
\epsilon\varPsi\sb 1
                                 \\
\varPsi\sb 2
\end{bmatrix}
\quad
\mbox{and}
\quad
\hat\varPhi\sb 2=\begin{bmatrix}\varphi_0\\0\end{bmatrix},
\]
we calculate the first diagonal term:
\begin{eqnarray}
\nonumber
&&
\epsilon^{-4}\Big\langle\p\sb{y\sp j}{}\phi,\chi_R\p\sb{y\sp k}{}
\phi\Big\rangle_{L^2}
\\
\nonumber
&&
=
\epsilon^2\Big\langle
\p\sb{y\sp j} \varphi_0,
\Big(-\frac{\Delta_y}{4m^2}\Big)\p\sb{y\sp k} \varphi_0
\Big\rangle_{L^2}
+
\Big\langle
\p\sb{y\sp j} \varphi_0,
\p\sb{y\sp k} \varphi_0
\Big\rangle_{L^2}\,+\,O(\epsilon)\,+o(1) 
\\
\nonumber
&&
=
\frac{\delta\sb{j k}}{3}
\Big\langle
\varphi_0,
(-\Delta_y)\varphi_0
\Big\rangle_{L^2}\, +\,O(\epsilon)\,+o(1),
\end{eqnarray}
where we took into account the spherical symmetry
of $\varphi_0$,
which leads to
\[
\langle \p\sb{y\sp 1} \varphi_0,\p\sb{y\sp 1}\varphi_0\rangle_{L^2}
=\frac 1 3\langle \varphi_0,(-\Delta_y)\varphi_0\rangle_{L^2}.
\]
The notation \(o(1) \) indicates the error term which is independent
of \(\epsilon \) and has limit zero as \(R\to+\infty \), and arises
from the limit of convergent integrals such as
\[
\Big\langle
\p\sb{y\sp j} \varphi_0,
\chi_R\p\sb{y\sp k} \varphi_0
\Big\rangle_{L^2}\,=\,
\Big\langle
\p\sb{y\sp j} \varphi_0,
\p\sb{y\sp k} \varphi_0
\Big\rangle_{L^2}\,+o(1)\,.
\]
Next the off-diagonal terms are \(O(\epsilon R)+o(1) \);
indeed, using the same expression for $\epsilon^{-2}\phi$ as above,
we compute:
\begin{eqnarray}
\nonumber
&&
\epsilon^{-4}
\Big\langle\p\sb{y\sp j}{}\phi,\chi_R(l_{k'}\phi+\frac i 2\Sigma\sb{k'}{}\phi)
\Big\rangle_{L^2}
\\
\nonumber
&&
=
-
\frac{\epsilon^2}{4m^2}
\Big\langle
\p\sb{y\sp j}
{\mom\,}
\begin{bmatrix}\varphi_0                  \\0\end{bmatrix},
\frac{i}{2}
\sigma\sb{k'}
{\mom\,}
\begin{bmatrix}\varphi_0                  \\0\end{bmatrix}
\Big\rangle_{L^2}
+
\Big\langle
\p\sb{y\sp j}\begin{bmatrix}\varphi_0     \\0\end{bmatrix},
\frac{i}{2}
\sigma\sb{k'}\begin{bmatrix}\varphi_0     \\0\end{bmatrix}
\Big\rangle_{L^2}
\\
\nonumber
&&
\qquad\quad\,
+\,O(\epsilon R)+o(1)\,.
\end{eqnarray}
The first two terms are actually identically zero 
since $\varphi_0$ is spherically symmetric
(so that by parity considerations 
it is $L^2$-orthogonal to all of its first partial derivatives, which
are in turn orthogonal to all of the second partial derivatives). The
\(O(\epsilon R) \) error term arises from the bound
\(\|\chi_R \bm{l}\varPsi_j\|_{L^2}\leq \const R\|\varPsi_j\|_{H^1} \),
etc.

Finally, for the second diagonal term:
\[\
\epsilon^{-4}
\Big\langle
(l_{j'}+\frac i 2\Sigma\sb{j'})\phi
,\chi_R(l_{k'}+\frac i 2\Sigma\sb{k'})
\phi
\Big\rangle_{L^2}
\,=\,
\frac{\delta\sb{j' k'}}{4}
\langle
\varphi_0,\varphi_0\rangle_{L^2}+O(\epsilon R^2)+o(1).
\]
(Recall that \(\varphi_0\) is radial so that \(l_j\varphi_0=0\) for each
\(j\).)
The nondegeneracy
of the matrix \eqref{xi-xi}  for large fixed \(R \) (again depending
only on \(\varphi_0 \)) and sufficiently 
small $\epsilon$ follows.
\end{proof}

Returning to the proof of Proposition~\ref{prop-sw-dm}, the above implies
that if \(R \) is fixed sufficiently large then there is an interval
\((-\epsilon\sb\ast,\epsilon\sb\ast) \) on which there is a solution
$\varPhi(y,\epsilon)$ of \(\mathcal{F}(\varPhi,\epsilon)=0 \).
Now the implicit function theorem proves that this solution
is $C^\infty$ as a function of $\epsilon\in(-\epsilon\sb\ast,\epsilon\sb\ast)$,
and so
\begin{equation}\label{phi1-phi2-epsilon-1}
\norm{{\varPhi}(\epsilon)-\hat\varPhi}\sb{H\sp 2}
=O(\epsilon)\,.
\end{equation}
To prove a stronger estimate
\eqref{phi1-phi2-epsilon-2},
we take the derivative of
\eqref{df}
with respect to $\epsilon$ at $\epsilon=0$;
this yields
\[
\eurM
\p\sb\epsilon\varPhi\at{\epsilon=0}
=0,
\]
with $\eurM$ given by \eqref{def-m}.
Due to \eqref{varphi-varpsi},
one has
$\p\sb\epsilon\varPhi\at{\epsilon=0}
=\p\sb\epsilon\varPsi\at{\epsilon=0}$;
the requirement
\eqref{varphi-perp-ker-m}
leads to 
$\p\sb\epsilon\varPhi\at{\epsilon=0}
=\p\sb\epsilon\varPsi\at{\epsilon=0}=0$,
and hence
\[
\norm{\varPhi(\epsilon)-\hat\varPhi}\sb{H^2}=O(\epsilon^2)\,.\]
Finally, notice that since the explicit dependence of $\mathcal{F}$
is on $\epsilon^2$, we have
\begin{equation}\label{pm}
\mathcal{F}
\big(\varPhi_1(-\epsilon),\varPhi_2(-\epsilon),-\epsilon\big)=
\mathcal{F}\big(\varPhi_1(-\epsilon),\varPhi_2(-\epsilon),+\epsilon\big)
=0
\end{equation}
and hence
$\varPhi_j(\epsilon)=\varPhi_j(-\epsilon)$, since otherwise it would
be possible to contradict
the local uniqueness part of the conclusion of the implicit function theorem
(applied to ${\mathcal{G}_R}$ with $\mathbf{a}=0$, $\mathbf{b}=0$).
This completes the proof
of Proposition~\ref{prop-sw-dm}
and thus of the existence part of Theorem~\ref{theorem-sw-dm}.
\end{proof}

\begin{remark}
The solutions of $\mathcal{F}(\varPhi,\epsilon)=0$ are obtained for 
both positive and negative epsilon close to zero, but
the $\epsilon$ negative branch apparently gives rise to 
solutions of the Dirac--Maxwell system via \eqref{ansatz-2}
which are related to the positive branch as follows.
By \eqref{ansatz-2},
the branch which corresponds
to negative $\epsilon$
has the form
\[
\tilde\phi(x,\omega)
=\begin{bmatrix}
\tilde\phi\sb 1(x,\omega)
                          \\
\tilde\phi\sb 2(x,\omega)
\end{bmatrix}
=\begin{bmatrix}
(-\epsilon)^3
\varPhi\sb 1(-\epsilon x,-\epsilon)
                          \\
(-\epsilon)^2
\varPhi\sb 2(-\epsilon x,-\epsilon)
\end{bmatrix}
=\begin{bmatrix}
-\epsilon^3
\varPhi\sb 1(-\epsilon x,\epsilon)
                          \\
\epsilon^2
\varPhi\sb 2(-\epsilon x,\epsilon)
\end{bmatrix},
\]
\[
\omega=-\sqrt{m^2-\epsilon^2},
\qquad
\epsilon\ge 0,
\]
where we took into account that
$\varPhi_1(y,\epsilon)$,
$\varPhi_2(y,\epsilon)$
obtained from Proposition~\ref{prop-sw-dm} are even in 
$\epsilon$.
Comparing to
\eqref{ansatz-2},
we conclude that
this branch is related to the
$\epsilon$-positive
branch $\phi(x,\omega)$
by
\[
\tilde\phi\sb 1(x,\omega)=-\phi\sb 1(-x,\omega),
\qquad
\tilde\phi\sb 2(x,\omega)=\phi\sb 2(-x,\omega),
\]
so that
$\tilde A\sp 0(x,\omega)=A\sp 0(-x,\omega)$,
\ $\tilde{\mathbf{A}}(x,\omega)=-\mathbf{A}(-x,\omega)$.
Consequently, these two branches
have
the same magnetic field but opposite
electric field (see \cite[\S2.3 and \S5.4]{MR0187641}).
\end{remark}
\begin{remark}
We briefly consider the symmetry properties of the solitary wave
solutions:
in
\cite[\S 2]{wakano-1966},
Wakano gives the Ansatz for the solitary waves in the cylindrical
coordinates $(\rho,z=r\cos\theta,\upphi)$, from which 
symmetry properties can be deduced. For our situation the
relevant Ansatz for the Dirac wave function is
\begin{equation}\label{cyl}
\phi(x)
=\begin{bmatrix}
\varphi_1(\rho,\theta)            \\
\varphi_2(\rho,\theta)e^{i\upphi} \\
i\varphi_3(\rho,\theta)            \\
i\varphi_4(\rho,\theta)e^{i\upphi} \\
\end{bmatrix}\,.
\end{equation}
An alternative approach to the existence theorem would be to 
set the problem up and then apply the
implicit function theorem entirely within this symmetry class.
The uniqueness assertion of the implicit function theorem would then
imply that the solutions so constructed agree with those
obtained above from Proposition~\ref{prop-sw-dm}.
\end{remark}
\begin{remark}
  The solution obtained is actually a convergent power series in \(\epsilon\) since
  all mappings involved are analytic and so the analytic implicit function theorem holds.
  \end{remark}
\begin{lemma}\label{lemma-a-decay}
There is $C<\infty$ such that
\[
\eurA\sp 0(y)
\,=\,\ech\frac{\norm{\varPhi\sb 2}^2+\epsilon^2\norm{\varPhi\sb 1}^2}
{4\pi|y|}\,+\,O(\langle y\rangle^{-2})
\,,
\qquad
\abs{\eubA(y)}\le C\langle y\rangle^{-2}.
\]
\end{lemma}

\begin{proof}
We just apply the multipole expansion \cite{MR614447}
to \eqref{sys-a}. The integrands are quadratic in the components of
the Dirac field \(\phi\in H^{2,\theta} \), and hence have exponential decay
\[
\sup_{y\in\R^3}\,e^{2\theta|y|}\,\biggl(
\Bigl|\varPhi\sb{2}\sp\ast{\varPhi\sb{2}}
+\epsilon^2\varPhi\sb{1}\sp\ast{\varPhi\sb{1}}\Bigr|\,+\,
\Bigl|\varPhi\sb{1}\sp\ast\bm\sigma{\varPhi\sb{2}}
+\varPhi\sb{2}\sp\ast\bm\sigma{\varPhi\sb{1}}\Bigr|
\biggr)\,<\,\infty\,.\]
This allows that \(\eurA\sp 0 \) has leading asymptotic
behaviour given by the Coulomb law  
\[
\eurA\sp 0(y)
\,=\,\ech\frac{\norm{\varPhi\sb 2}^2+\epsilon^2\norm{\varPhi\sb 1}^2}
{4\pi|y|}\,+\,O(\langle y\rangle^{-2})
\]
as \(|y|\to+\infty\,. \) The vector potential \(\eubA \)
however has no monopole component because the currents \(\mathbf{J} \)
have zero integral when evaluated on any stationary solution which
decays rapidly at spatial infinity. Indeed for a stationary solution
the conservation law \(\partial_\mu J^\mu=0 \) implies that
\(\mathbf{J} \) is divergence-free, and hence:
$$
0
=\p\sb t\int\sb{\R^3} J\sp{0}y\sp k\,dy
=-\int\sb{\R^3}(\p\sb j J\sp j)y\sp k\,dy
=
\int\sb{\R^3} J\sp k\,dy,
\qquad
1\le k\le 3.
$$
As a consequence,
the multipole expansion implies that
\(\eubA=O(\langle y\rangle^{-2}) \).
\end{proof}

\section*{Acknowledgments}
The research of Andrew Comech was carried out
at the Institute for Information Transmission Problems
of the Russian Academy of Sciences
at the expense of the Russian Foundation
for Sciences (project 14-50-00150).
The work of David Stuart has been partially supported by STFC consolidated grant ST/P000681/1 and St John's College, Cambridge.

\bibliographystyle{sima-doi}
\bibliography{bibcomech}
\end{document}